\documentclass[pra,reprint,twocolumn,superscriptaddress,floatfix,nofootinbib]{revtex4-2}
\pdfoutput=1

\usepackage[utf8]{inputenc}
\usepackage[T1]{fontenc}
\usepackage[colorlinks]{hyperref}
\usepackage[normalem]{ulem}
\usepackage{comment}

\usepackage{amsthm}
\usepackage{mathtools}
\usepackage{bbm}
\usepackage{bm}
\usepackage{dsfont}
\usepackage{graphicx}
\usepackage{longtable}
\usepackage{booktabs} 
\usepackage{algorithm}
\usepackage{algpseudocode}
\usepackage{amssymb}

\usepackage{physics}
\usepackage{tikz}
\usetikzlibrary{arrows}

\usepackage{thmtools}
\usepackage{thm-restate}

\newtheorem{theorem}{Theorem}
\newtheorem{lemma}[theorem]{Lemma}

\newtheorem{proposition}[theorem]{Proposition}
\theoremstyle{definition}

\newcommand{\id}{\mathds{1}}
\renewcommand{\H}{\mathcal{H}}
\renewcommand{\L}{\mathcal{L}}
\newcommand{\U}{\mathcal{U}}
\newcommand{\M}{\mathcal{M}}
\renewcommand{\O}{\mathcal{O}}

\newcommand{\C}{\mathcal{C}}

\newcommand{\K}{\mathcal{K}}

\renewcommand{\S}{\mathcal{S}}
\newcommand{\T}{\mathcal{T}}
\newcommand{\V}{\mathcal{V}}
\newcommand{\W}{\mathcal{W}}

\def\dket#1{\mathinner{|{#1}\rangle\!\rangle}}
\def\dketbra#1#2{\mathinner{|{#1}\rangle\!\rangle\!\langle\!\langle{#2}|}}

\begin{document}


\title{Quantum Query Complexity of Boolean Functions under Indefinite Causal Order}

\author{Alastair A.\ Abbott}
\affiliation{Univ.\ Grenoble Alpes, Inria, 38000 Grenoble, France}

\author{Mehdi Mhalla}
\affiliation{Univ.\ Grenoble Alpes, CNRS, Grenoble INP, LIG, 38000 Grenoble, France}

\author{Pierre Pocreau}
\affiliation{Univ.\ Grenoble Alpes, Inria, 38000 Grenoble, France}
\affiliation{Univ.\ Grenoble Alpes, CNRS, Grenoble INP, LIG, 38000 Grenoble, France}

\date{October 26, 2023}


\begin{abstract}
The standard model of quantum circuits assumes operations are applied in a fixed sequential ``causal'' order. In recent years, the possibility of relaxing this constraint to obtain causally indefinite computations has received significant attention. The quantum switch, for example, uses a quantum system to coherently control the order of operations. Several \textit{ad hoc} computational and information-theoretical advantages have been demonstrated, raising questions as to whether advantages can be obtained in a more unified complexity theoretic framework.
In this paper, we approach this problem by studying the query complexity of Boolean functions under general higher order quantum computations. To this end, we generalise the framework of query complexity from quantum circuits to quantum supermaps to compare different models on an equal footing. 
We show that the recently introduced class of quantum circuits with quantum control of causal order cannot lead to any reduction in query complexity, and that any potential advantage arising from causally indefinite supermaps can be bounded by the polynomial method, as is the case with quantum circuits. Nevertheless, we find some functions for which the minimum error with which they can be computed using two queries is strictly lower when exploiting causally indefinite supermaps.
\end{abstract}


\maketitle


\textit{Introduction.---}%
The query model of computation is a simplified framework that offers a way to evaluate the complexity of different computational problems and to probe and compare different computational frameworks. 
Notably, it has been used to prove various separations between classical and quantum computations, and to provide non-trivial lower bounds on some important computational problems, such as the optimally of Grover's algorithm or the complexity of graph problems~\cite{bennett_strengths_1997, grover_fast_1996,Durr_2006}. 
This model was thus far only studied in the context of quantum circuits, in which queries and unitary operations are performed in a fixed sequential order, implying an underlying fixed causal structure. 
Several recent works have proposed ways of computing beyond fixed causal structures using, for example, the quantum switch~\cite{chiribella_quantum_2013}, where a quantum control system can put different sequences of operations into superposition, leading to what is called an indefinite causal order. 
Such scenarios can be studied in the higher order framework of quantum supermaps~\cite{chiribella_quantum_2013} (or ``process matrices''~\cite{Oreshkov2012}) which describes the most general transformations taking some given input operations into a valid output operation in a consistent manner.
Although some quantum supermaps are compatible with a fixed causal order, and thus represent quantum circuits~\cite{chiribella_quantum_2008}, others, like the quantum switch, have indefinite causal order or, more formally, are \emph{causally nonseparable}~\cite{Oreshkov2012,oreshkov_causal_2016,Wechs_2019}.

Causal indefiniteness was shown to provide advantages over quantum circuits in several \emph{ad hoc} tasks in quantum information theory~\cite{PhysRevA.86.040301, PhysRevLett.113.250402, PhysRevLett.128.230503, PRXQuantum.2.010320, bavaresco2022unitary, quintino2022deterministic, PhysRevLett.124.190503}.
However, much less is known about its potential advantages in a general complexity theoretic sense except that its computational complexity is upper bound by PP (which includes both NP and BQP)~\cite{PhysRevA.96.052315}.

In this letter, we address this problem by proposing query complexity as a natural way of comparing the relative power of quantum circuits and causally indefinite supermaps. 
After recalling the standard query model of computation in terms of quantum circuits, we formulate three different families of quantum supermaps that can be understood as query computations with different causal structures.
We show that allowing for superpositions or coherent control of the order of queries in so-called \emph{quantum circuits with quantum control of causal order}~\cite{wechs_quantum_2021} cannot reduce the query complexity of Boolean functions compared to quantum circuits. 
Furthermore, through a generalisation of the polynomial method to general quantum supermaps, we give a lower bound on any possible reduction that could be obtained with causally indefinite supermaps. 
Then, using semidefinite programming techniques, we show that an advantage can be obtained with some such supermaps: for some functions, causal indefiniteness can reduce the minimum error with which they can be computed given a fixed number of queries.


\textit{The query model of computation.---}%
In the quantum query model of computation (see~\cite{BUHRMAN200221} or \cite{ambainis2018understanding} for a survey), the goal is to compute a Boolean function $f : \{0, 1\}^n \to \{0, 1\}$ with a quantum circuit having access to an oracle $O_x$.
Here without loss of generality we use the phase oracle~\cite{ambainis_quantum_2002,ambainis2018understanding}, defined as the unitary that takes as input a state $\ket{i}$ in an $(n+1)$-dimensional Hilbert space $\H^Q$, with $i$ serving as an index, $0 \leq i \leq n$ such that:
\begin{equation}
\label{eq:phase}
    O_x\ket{i} = \begin{cases} (-1)^{x_i}\ket{i} & \text{ if } i \not = 0, \\
                                            \ket{i} & \text{ otherwise.} \end{cases}
\end{equation}
Note that the special case $\ket{i} = \ket{0}$ is needed to ensure $O_x$ is equivalent to the more common form oracle $O'_x\ket{j, b} = \ket{j, b \oplus x_j}$, with $1\leq j\leq n$, $b \in \{0, 1\}$ and $\oplus$ denoting addition modulo 2.
A quantum circuit with $T$ queries to the oracle is specified by a series of $T+1$ unitaries $(U_1, \dots, U_{T+1})$ each acting on the space $\H^{Q\alpha} = \H^Q \otimes \H^\alpha$, with $\H^\alpha$ an ancillary space of arbitrary dimension. 
By extending the oracle on the ancillary space, $\tilde{O}_x = O_x \otimes \id^\alpha$, we say the circuit computes $f$ (exactly) if for all $x$, measuring the left most qubit of $U_{T+1} \tilde{O}_x U_{T} \cdots U_1 \tilde{O}_x U_1 \ket{0 \cdots 0}$ gives the outcome $f(x)$ with probability one. 
We similarly say that the circuit computes $f$ with a bounded error $\varepsilon$ if the measurement result gives, for all $x$, $f(x)$ with a probability at least $1-\varepsilon$ with some fixed $\varepsilon < 1/2$; typically one takes $\varepsilon=1/3$. 
The exact quantum query complexity of a Boolean function, written $Q_E(f)$, is the minimal value of $T$ for which there exists a quantum query circuit computing $f$ exactly, while $Q_2(f)$ is the (two-sided) bounded error quantum query complexity, defined as the minimal $T$ for which a quantum query circuit exists computing $f$ with bounded error $\varepsilon = 1/3$. 

This definition of quantum query complexity is the usual one considered in the literature, and is  based on the standard formalism of quantum circuits. In this letter, we show that translating this definition into the framework of quantum supermaps leads to a natural generalisation of the notion of quantum query complexity to study the relative power of causal indefiniteness. Let us first introduce the tools and notation we will use throughout this letter.


\textit{Mathematical tools and notation.---}%
We denote the space of linear operators on a Hilbert space $\H^X$ as $\L(\H^X)$.
To define quantum supermaps we first introduce \emph{quantum channels}, which are completely positive (CP) trace-preserving (TP) maps $\M : \L(\H^X) \to \L(\H^Y)$.
A channel is conveniently represented as a positive semidefinite matrix $\mathsf{M}^{XY} \in \L(\H^{XY})$ using the Choi isomorphism~\cite{CHOI1975285,JAMIOLKOWSKI1972275} (see Appendix~\ref{appendix:Choi}), where we adopt the shorthand $\H^{XY}=\H^X\otimes\H^Y$. 
The composition of two CP maps, potentially over a subset of their input/output systems, can be computed directly via their Choi matrices using the ``link product''~\cite{chiribella_theoretical_2009, chiribella_quantum_2008}, denoted `$*$' and defined for any matrices $\mathsf{M}^{XY}\in\L(\H^{XY})$ and $\mathsf{N}^{YZ} \in \L(\H^{YZ})$ as $\mathsf{M}^{XY} * \mathsf{N}^{YZ} = \Tr_{Y}[(\mathsf{M}^{XY} \otimes \id^Z)^{\mathsf{T}_Y} (\id^X \otimes \mathsf{N}^{YZ})] \in \L(\H^{XZ})$, where $\Tr_Y$ is the partial trace over the subsystem $\H^Y$ and $\cdot^{\mathsf{T}_Y}$ the partial transpose over $\H^Y$. 
The link product is commutative and associative and reduces to the tensor product on maps acting on disjoint Hilbert spaces (see Appendix~\ref{appendix:link_prod}). 

A $T$-input \emph{quantum supermap} $\mathcal{S}$ is then a $T$-linear, completely CP-preserving, and TP-preserving map~\cite{chiribella_transforming_2008}.
That is, a supermap is the most general transformation that transforms any $T$ input channels $\M_t: \L(\H^{A_k^I}) \to \L(\H^{A_k^O})$ into another quantum channel: $\mathcal{S}(\M_1, \dots, \M_T) : \L(\H^P) \to \L(\H^F)$.
Writing $\H^{A^{IO}_k}=\H^{A^{I}_kA^O_k}$ and, with $\T := \{1, \dots, T\}$, $\H^{A^{IO}_\T} = \bigotimes_{k \in \T} \H^{A^{IO}_k}$, a supermap can be represented in the Choi picture as a ``process matrix''~\cite{Oreshkov2012}, a positive semidefinite matrix $W \in \L(\H^{P A^{IO}_\T F})$ belonging to a specific subspace $\L^{\text{Gen}}$ (see Appendix~\ref{annex:Gen}) and normalised such that $\Tr W = d^P \prod_{k=1}^T d^O_k$, with $d^P = \dim(\H^P)$ and $d^O_k = \dim(\H^{A^O_k})$~\cite{araujo_witnessing_2015,Wechs_2019}. 
We denote the set of all such process matrices $\W^{\text{Gen}}$.
The process matrix $W$ fully characterises $\mathcal{S}$, and the Choi matrix of the channel resulting from applying $\mathcal{S}$ to some input channels $(\mathcal{M}_1, \dots, \mathcal{M}_T)$ is obtained as $\mathsf{S}(\mathsf{M}_1, \dots, \mathsf{M}_T) = (\mathsf{M}_1 \otimes \dots \otimes \mathsf{M}_T) * W \in \L(\H^{PF})$.

\emph{Fixed order and quantum controlled supermaps.---}%
The general framework of quantum supermaps allows us to compare, on an equal footing, specific classes of quantum supermaps of particular interest, and to understand their computational capabilities.
Here we present two sub-classes of particular interest: supermaps with a fixed causal order (sometimes termed quantum combs \cite{chiribella_quantum_2008} and which are equivalent to quantum circuits), and quantum supermaps with quantum control of causal order~\cite{wechs_quantum_2021}.

The class of supermaps compatible with a fixed causal order (or \emph{FO-supermaps}) has been extensively studied.
Any such $T$-input supermap $\mathcal{S}$ can de described by $T+1$ channels $\V_1 : \L(\H^P) \to \L(\H^{A^I_{1} \alpha_1})$, $\V_t : \L(\H^{A^O_t \alpha_t}) \to \L(\H^{A^I_{t+1} \alpha_{t+1}})$ for $2 \leq t \leq T$, and $\V_{T+1} : \L(\H^{A^O_{T} \alpha_T}) \to \L(\H^F)$ composed sequentially in a circuit structure~\cite{chiribella_quantum_2008}.
Here, the $\H^{\alpha_t}$ are an ancillary spaces, and $\H^P$ and $\H^F$ are the input and output spaces of the circuit.
In the Choi picture, the action of $\mathcal{S}$ on $T$ input channels $(\M_1, \dots, \M_T)$ is
\begin{align}
\label{eq:ActionFO}
    \mathsf{S}(\mathsf{M}_1, \dots, \mathsf{M}_T) =& \mathsf{V}_{T+1} * \mathsf{M}_T * \mathsf{V}_{T} * \dots * \mathsf{V}_2 *  \mathsf{M}_1 * \mathsf{V}_1 \notag \\
    =&(\mathsf{M}_1 \otimes \dots \otimes \mathsf{M}_T) * W \in \L(\H^{PF}),
\end{align}
where $W = \mathsf{V}_{T+1} * \dots * \mathsf{V}_1 \in \L(\H^{PA^{IO}_{\mathcal{T}}F})$ is the process matrix of $\mathcal{S}$. Process matrices of this specific form---which can be decomposed as the link product of quantum channels---belong to a subset $\mathcal{W}^{\text{FO}}\subset\mathcal{W}^\text{Gen}$~\cite{chiribella_quantum_2008,chiribella_theoretical_2009, wechs_quantum_2021}  (see Appendix~\ref{appendix:FO}). 
It is easy to see that quantum query circuits, as defined above, can readily be represented as FO-supermaps.

Another family of supermaps, those with quantum control of causal order (or \emph{QC-supermaps}), can be formulated as generalised quantum circuits in which a quantum system controls the order of application of the different input operations~\cite{wechs_quantum_2021}.
This clear interpretation and the potential realizability of QC-supermaps~\cite{wechs_quantum_2021,purves21,salzger23} makes it a class of particular interest, and which includes supermaps such as the quantum switch that are causally indefinite~\cite{chiribella_quantum_2013}.

A QC-supermap alternates between applying some controlled ``internal'' operations on a ``target'' system and some ancillary systems, potentially correlating these systems with the quantum control, and using this control system to control which input operations to apply to the target system at time-step $t$.
More precisely, throughout the computation the control system, at time-step $t$, has basis states $\ket{\K_{t-1},k_t}$ specifying that operation $k_t\notin \K_{t-1}$ will be applied next and that the operations $\K_{t-1}\subsetneq\T$ have already been applied.
This system controls coherently the application of the different input operations $\mathcal{M}_{k_t}$ at time-step $t$ and the pure internal operations $\V_{\K_{t-1},k_t}^{\to k_{t+1}} : \L(H^{A^O_{k_t}\alpha_t}) \to \L(H^{A^I_{k_{t+1}}\alpha_{t+1}})$ between time-steps, and is the minimal control system required to ensure that no operation is applied more than once, while maintaining the possibility for superpositions of causal orders and interference of different causal histories (i.e., permutations of the operations in $\K_{t-1}$); see Appendix~\ref{appendix:QC} or~\cite{wechs_quantum_2021} for a more detailed description of QC-supermaps.
Process matrices of QC-supermaps belong to a subset $\mathcal{W}^{\text{QC}}\subset \mathcal{W}^\text{Gen}$~\cite{wechs_quantum_2021}, whose characterisation is also given in Appendix~\ref{appendix:QC}. 


\textit{Quantum query complexity under indefinite causal order.---}%
We can now define what it means for a quantum supermap to compute a Boolean function $f : \{0, 1\}^n \to \{0, 1\}$,
in order to compare the relative computational power of the different types of supermaps. 

Let us consider a $T$-input supermap $\mathcal{S}^{\mathcal{C}}$ of class $\mathcal{C} \in \{\text{FO, QC, Gen} \}$, with trivial input space $\dim(\H^P) = 1$ and output space of dimension $\dim(\H^F) = 2$. 
These dimensional constraints reflect the fact that, as in the initial definition of quantum query complexity, any fixed input to the supermap can be absorbed directly into it, and only a single output qubit---that will be measured to define the output of the computation---is needed.

To differentiate the different input and output spaces of each query, we label (without imposing a particular order) the $T$  queries to the phase oracle~\eqref{eq:phase} as $\O_x^{(1)},\dots,\O_x^{(T)}$, with $\O_x^{(i)}: \L(\H^{A^I_i}) \to \L(\H^{A^O_i})$, where $\H^{A^I_i}$ and $\H^{A^O_i}$ are of dimension $n+1$.
(Here $\mathcal{O}_x$ is the quantum channel corresponding to the unitary oracle $O_x$.) 
Then for every $x$, the output of $\mathcal{S}^{\mathcal{C}}$ is the qubit state $\mathcal{S}^{\mathcal{C}}(\O_x^{(1)}, \dots, \O_x^{(T)}) \in \L(\H^F)$.
The probability that one obtains $f(x)$ when measuring this qubit in the computational basis is then given by the Born rule as
\begin{align}
\label{eq:computation}
    p\big(f(x)\big) &= \Tr[\mathsf{S}^{\mathcal{C}}(\mathsf{O}_x^{(1)}, \dots, \mathsf{O}_x^{(T)}) \cdot \Pi_{f(x)}] \\
	&= \Tr[(\mathsf{O}_x^{\otimes T} * W^{\mathcal{C}})\cdot \Pi_{f(x)}]\\
    &= \Tr[(\mathsf{O}_x^{\otimes T} \otimes \Pi_{f(x)})\cdot W^{\mathcal{C}}],
\end{align}
with $\Pi_{f(x)} = \ketbra{f(x)}{f(x)}$ and $\mathsf{O}_x^{\otimes T} = \bigotimes_{i=1}^T \mathsf{O}_x^{(i)}$.  
We say that $\mathcal{S}^\mathcal{C}$ computes $f$ with a bounded error $\varepsilon$ if for all $x$, $p(f(x)) \geq 1- \varepsilon$.
We can now define $Q_E^{\mathcal{C}}(f)$ as the minimum value $T$ for which there exists a quantum supermap $\mathcal{S}^{\mathcal{C}}$ computing $f$ exactly (i.e., with error $\varepsilon=0$).
Similarly, the corresponding bounded-error complexity $Q_2^{\mathcal{C}}(f)$ is the minimum $T$ for which there exists a quantum supermap $\mathcal{S}^{\mathcal{C}}$ computing $f$ with bounded error $\varepsilon = 1/3$.
By definition we have the following inequalities:
\begin{align}
     Q^\text{Gen}_E(f) &\leq Q^\text{QC}_E(f) \leq Q^\text{FO}_E(f) = Q_E(f), \label{ineq:exact} \\
     Q^\text{Gen}_2(f) &\leq Q^\text{QC}_2(f) \leq Q^\text{FO}_2(f) = Q_2(f). \label{ineq:approx}
\end{align}
To compare the power of the different classes of quantum supermaps in terms of query complexity, we need to understand how these inequalities can be refined.


\textit{Quantum control over $T$ copies of the same unitary channel.---}%
QC-supermaps are known to provide advantages over FO-supermaps in several tasks, some of which are sometimes expressed as query complexity advantages (albeit of a different kind than studied here), such as in the discrimination between pairs of commuting and anti-commuting unitaries~\cite{PhysRevA.86.040301} and generalisations of that task~\cite{PhysRevLett.113.250402,PRXQuantum.2.010320}. 
Here, we show that they cannot provide an advantage in query complexity of Boolean functions. In particular, building on a result of Ref.~\cite{bavaresco2022unitary} for ``swich-like'' supermaps, as strict subclass of QC-supermaps generalising the quantum switch, we show that QC-supermaps are equivalent to FO-supermaps when acting on $T$ copies of the same unitary channel.

\begin{restatable}{thm}{QC}
\label{th:QC-S}
 For any $T$-input QC-supermap $\S^\textup{QC}$, there exists an FO-supermap $\S^\textup{FO}$ such that $\S^\textup{QC}$ and $\S^\textup{FO}$ have the same action whenever applied to $T$ copies of the same unitary channel. That is, for all unitary channels $\mathcal{U}$, $\S^\textup{QC}(\mathcal{U},\dots,\mathcal{U})=\S^\textup{FO}(\mathcal{U},\dots,\mathcal{U})$.
\end{restatable}
The proof of Theorem~\ref{th:QC-S} is detailed in Appendix~\ref{appendix:equiv}. 
While the theorem may appear intuitive, the result exploits subtly both the specific structure of QC-supermaps and the unitarity of the operation. Indeed, the analogous result for non-unitary channels is known to be false, and concrete advantages are known in that setting~\cite{quintino19,PhysRevLett.127.200504,liu23,mothe23}.

A direct corollary of Theorem \ref{th:QC-S} is that $Q_E^\text{QC}=Q_E^\text{FO}$ and $Q_2^\text{QC}=Q_2^\text{FO}$.
Thus, no advantage in query complexity can be found using QC-supermaps instead of FO-supermaps, and for that matter, in any task where $T$ copies of the same unitary channel are considered, such as the reversal of unknown unitary transformations~\cite{quintino19} or unitary channel discrimination~\cite{bavaresco2022unitary}. 
In both these tasks, however, advantages were obtained using more general causally indefinite quantum supermaps beyond QC-supermaps.
This raises the prospect of nevertheless obtaining advantages in query complexity from causally indefinite supermaps.
In order to better target where to look for such an advantage, we first provide a lower bound on any potential reduction of query complexity of Boolean functions with general supermaps by generalising a well-studied bound for quantum circuits, the polynomial method.


\textit{Polynomial bound for general quantum supermaps.---}%
The polynomial method makes a connection between the output of an FO-supermap and a multivariate polynomial $g$, and is an important method for proving lower bounds on the quantum query complexity of Boolean functions~\cite{beals_quantum_2001}. 
A polynomial $g$ is said to represent a Boolean function $f$ if for all $x \in \{0, 1\}^n$, $f(x) = g(x)$. As an example, consider the polynomial $g(x_1, x_2) = x_1 + x_2 - x_1x_2$, which represents the Boolean $\text{OR}$ function. This polynomial has degree $2$, and we denote $\deg(f)$ the smallest degree of any polynomial representing $f$.
Similarly, we write $\widetilde{\deg}(f)$ the degree of the smallest polynomial approximating $f$ with a bounded error $\varepsilon=1/3$, i.e., such that for all $x$, $|g(x) - f(x)| \leq 1/3$. 
The polynomial method states that for all Boolean functions we have $\deg(f)/2 \leq Q^\text{FO}_E(f)$, and likewise for the bounded-error counterpart, that $\widetilde{\deg(f)}/2 \leq Q^\text{FO}_2(f)$. 
Here, we generalise this lower bound to general supermaps and thereby bound the potential advantage obtainable with causally indefinite supermaps over standard quantum circuits.
\begin{restatable}{thm}{boundProcess}
\label{th:bound_process}
  For any Boolean function $f$, we have $\deg(f)/2 \leq Q^\textup{Gen}_E(f)$ and $\widetilde{\deg}(f)/2 \leq Q^\textup{Gen}_2(f)$.
\end{restatable}
The proof is similar to that of the original result; we provide the full details in Appendix~\ref{appendix:TheoremPol} and simply outline the argument here. 
Supposing that $Q_E^{\text{Gen}}(f)=T$, there is a $T$-input supermap $\S^{\text{Gen}}$ with process matrix $W^{\text{Gen}}$ such that $f(x)=\Tr\big[(\mathsf{O}_x^{\otimes T} * W^{\text{Gen}}) \cdot \Pi_1\big]$.
Proceeding by induction, one finds that the Choi matrix of the $T$ queries, $\mathsf{O}_x^{\otimes T}$, has elements that are multivariate polynomials of degree at most $2T$.
Since both the trace and link product are linear, it immediately follows that $f(x)$ is also a multivariate polynomial of degree at most $2T$, completing the proof.
An analogous proof gives us the bounded error version of the statement.

These results mean that for functions whose polynomial bound is tight for FO-supermaps, i.e., when $\deg(f)/2 =  Q^{\text{FO}}_E(f)$ or $\widetilde{\deg}(f)/2 =  Q^{\text{FO}}_2(f)$, causal indefiniteness cannot provide any advantage.
This is the case, for example, of the OR function~\cite{beals_quantum_2001} (which is computed with a bounded error by Grover's algorithm). 
Theorems~\ref{th:QC-S} and~\ref{th:bound_process} thus allow us to refine the inequalities~\eqref{ineq:exact} and~\eqref{ineq:approx} as follows:
\begin{align}
     \frac{\deg(f)}{2} &\leq Q^\text{Gen}_E(f) \leq Q^\text{QC}_E(f) = Q^\text{FO}_E(f) = Q_E(f) \\
    \frac{\widetilde{\deg}(f)}{2} &\leq Q^\text{Gen}_2(f) \leq Q^\text{QC}_2(f) = Q^\text{FO}_2(f) = Q_2(f).
\end{align}
Still, it is known that some Boolean functions do not have a tight polynomial bound \cite{AMBAINIS2006220}, meaning that Theorem~\ref{th:bound_process} does not rule out a potential advantage from causal indefiniteness using general supermaps. 
This motivates us to study explicit Boolean functions to look for such an advantage.
We note that some bounds can be immediately put on the extent of any possible separation: since $Q_E(f)=\tilde{O}(\deg(f)^3)$ and $Q_2(f)=\tilde{O}(\widetilde{\deg}(f)^4)$~\cite{aaronson16,aaronson2021degree}, one cannot hope for an exponential advantage from causal indefiniteness.\footnote{A function $f(x)=\tilde{O}(g(x))$ if $f(x)=O(g(x)\log^k x)$ for some constant $k$.}

In general, we lack an understanding of general supermaps beyond QC-supermaps, so trying to develop analytically supermaps providing such an advantage for general $n$-bit Boolean functions is extremely challenging.
Instead, we study exhaustively Boolean functions up to $4$ bits, building on the extensive literature on the optimisation of quantum supermaps using semidefinite programs (SDPs)~\cite{araujo_witnessing_2015,Wechs_2019}.


\textit{Advantage of general supermaps over FO-supermaps.---}%
In the remainder of this letter, we show an advantage of causally indefinite supermaps over fixed order ones by considering the minimum bounded error computation of a function $f$. 
Let us denote by $\varepsilon^\text{FO}_T(f)$ (resp.\ $\varepsilon^\text{Gen}_T(f)$) the minimum error $\varepsilon$ for which there exists a $T$-query FO-supermap (resp.\ a general supermap) computing $f$ with a bounded error $\varepsilon$. We prove the following theorem for $T=2$.
\begin{restatable}{thm}{gap}
\label{theorem:gap}
There exists some function $f$ for which $\varepsilon^\textup{Gen}_2(f) < \varepsilon^\textup{FO}_2(f)$.
\end{restatable}

To prove Theorem~\ref{theorem:gap}, we use SDPs to look for such a gap on all Boolean functions up to 4 input bits. 
A systematic study for quantum circuits of the minimum bounded error computation of functions using SDPs was already performed in Ref.~\cite{Montanaro2015} for functions up to 4 bits, and for symmetric functions up to 6 bits. 
Their SDP formulation keeps track of the evolution of the state through the quantum circuit by defining $T$ Gram matrices, one for each additional query~\cite{QuerySDP}. 
Unfortunately this methods fails to extend to general supermaps as it exploits the sequential structure of a FO-supermap which has no analogue in generic quantum supermaps.\footnote{Indeed, for the same reason a more powerful method of bounding query complexity for quantum circuits, the adversarial method~\cite{ambainis_quantum_2002, hoyerNegativeWeightsMake2007a}, is not readily generalisable to generic supermaps, in contrast to the polynomial bound.}
Here, using the characterisation of supermaps in the Choi picture as process matrices $W$, we give an SDP formulation of the minimum error bound for both FO-supermaps and general ones.

Let us consider a Boolean function $f$ and a $T$-query supermap $\S^{\mathcal{C}}$ with $\mathcal{C}\in\{\text{FO},\text{Gen}\}$. 
Because we measure the output of the supermap to obtain the value of $f(x)$, it is convenient to consider, instead of the process matrix $W$, the quantum superinstrument $\{W^{[i]}\}_{i \in \{0, 1\}}$ such that $W^{[i]} = W * \Pi_{i}$~\cite{wechs_quantum_2021}, so that $p(f(x))=\Tr[W^{[f(x)]} \cdot \mathsf{O}_x^{\otimes T}]$. 
This has the advantage of reducing the size of the matrices being optimised numerically; for $T=2$ and $n=4$, $W$ is a matrix of size 1250, whereas the $W^{[i]}$ are of size 625, a difference crucial for rendering the SDPs we present below tractable.
Writing $F^{[i]} = \{x : f(x) = i\}$, the minimum error bound $\varepsilon^{\mathcal{C}}_T(f)$ is given by the SDP
\begin{equation}
\label{sdp}
\begin{split}
 \varepsilon^\mathcal{C}_T(f) &= \max_{\varepsilon, W^{[0]}, W^{[1]}} 1-\varepsilon \\
\text{s.t.}\ & \forall x \in F^{[0]}, \ \Tr[W^{[0]} \cdot \mathsf{O}_x^{\otimes T}] \geq 1-\varepsilon, \\
       & \forall x \in F^{[1]}, \ \Tr[W^{[1]} \cdot \mathsf{O}_x^{\otimes T}] \geq 1-\varepsilon, \\
       & W^{[0]} \geq 0,\ W^{[1]} \geq 0, \ \varepsilon \geq 0, \\
       & W^{[0]} + W^{[1]} \in \mathcal{W}^\mathcal{C}.
\end{split}
\end{equation}
Note that the size of the $W^{[i]}$'s scale as $(n+1)^{2T}$ making  the SDP difficult to solve for even moderate $n$ and $T$. 
We solved it numerically for 2 queries and for all Boolean functions up to 4 bits (including constant functions and functions depending on fewer bits). By exploiting the symmetries inherent in the SDP, one can further reduce the number of variables and constraints, but we were still unable to solve it for 3 queries or 5 bits. 

For 3 bits and 2 queries it is known that FO-supermaps are sufficient to exactly compute all functions except the AND function~\cite{Montanaro2015} which has an exact query complexity $Q^\text{FO}_E(\text{AND}) = n$. For this function, no advantage was found with 2-query general supermaps.

For 4-bit functions, as in Ref.~\cite{Montanaro2015}, we reduced the number of functions to look at by considering the so-called negate, permute, negate (NPN) equivalence relation. 
We say that two Boolean functions are equivalent if they are the same up to negation and permutation of the input bits, and negation of the output. 
These transformations correspond to a relabelling of the input and output bits, leaving unchanged the query complexity of a function.
 
The results are summarised in Table~\ref{table:results} of Appendix~\ref{app:numerical} where the ID of a function corresponds to its truth table converted into an integer, and the results of the SDPs for both FO- and general supermaps are rounded to the fifth decimal place. Our results for FO-supermaps coincide, as expected, with those of Ref.~\cite{Montanaro2015}.
Out of the 222 NPN representatives, we observe a gap between the minimum error bounds $\varepsilon^\textup{FO}_2(f)$ and  $\varepsilon^\textup{Gen}_2(f)$ for 179 functions, the maximum gap being 0.00947 (close to 1\%) for the functions with IDs 5783, 5865 and 6630.
The simplest of the three functions (ID 5865), can be written as the polynomial $f(x) = x_1 + x_2x_3 + x_2x_4 + x_3x_4 + x_2x_3x_4$.

This numerical evidence is not an analytical proof of Theorem~\ref{theorem:gap} as the constraints of the SDPs are only satisfied up to numerical precision. 
However, extracting an FO- or  general supermap that rigorously verifies the constraints is possible by rationalising the numerical results and perturbing the SDPs' solutions. Such a method was developed in Ref.~\cite{PhysRevLett.127.200504} and allows us to obtain upper bounds with the primal SDP~\eqref{sdp}, and lower bounds with the corresponding dual SDP which can be derived with Lagrangian method (see Appendix~\ref{Appendix:dual}). 
Using the extraction methods (detailed in Appendix~\ref{appendix:extraction}) for the function ID 5865, we find the bounds
\begin{equation}
    0.0324 \leq \varepsilon^\text{Gen}_2(f) \leq 0.0377 < 0.0465 \leq \varepsilon^\text{FO}_2(f) \leq 0.0467,
\end{equation}
which proves Theorem~\ref{theorem:gap}.


\textit{Discussion.---}%
By generalising the notion of query complexity to general supermaps, we provide a natural tool to probe, in a unified complexity theoretic framework and on equal footing, the power and the potential advantages that different types of causal structure can provide. 
Until now, it was unclear whether causal indefiniteness could provide any advantage in such a fundamental model of computation.
We found a separation between FO-supermaps and general, potentially causally indefinite, supermaps in the minimum error probability with which they can compute a Boolean function using two queries. 
While this separation does not directly translate into an asymptotic query complexity separation, it is a crucial first step in this direction.
One possible such approach would be to explore whether the separation we found can be amplified by recursively composing supermaps in a suitably well-defined manner~\cite{guerin19,kissinger19}, similar to the advantage in query complexity obtained with quantum circuits in Ref.~\cite{AMBAINIS2006220}. 
These results also raise the question of whether it is possible to find a separation in exact query complexity.

Our proof that QC-supermaps cannot provide any advantage over FO-supermaps when a single unitary is repeatedly queried has implications beyond query complexity. 
For instance, it also implies that this important class of supermaps cannot provide advantages in quantum metrology of unitary dynamics~\cite{giovannetti11,zhao}.
It likewise raises natural questions about the power of specific classes of supermaps beyond QC-supermaps, such as that of purifiable processes~\cite{Araujo2017purification}. 
The prospect that some such supermaps may be realizable~\cite{wechs23} leaves open the possibility of exploiting causal indefiniteness in the standard query complexity setting.
It would similarly be interesting to study whether supermaps, including QC-supermaps, can be exploited in related scenarios.
For instance, can advantages over quantum circuits be obtained when given multiple different oracles, and can our results be used to obtain advantages in quantum communication complexity in more standard settings than those studied with causal indefiniteness previously~\cite{guerin16}?


\begin{acknowledgments}
The authors acknowledge funding from the French National Research Agency projects ANR-15-IDEX-02 and ANR-22-CE47-0012, and the Plan France 2030 projects ANR-22-CMAS-0001 and ANR-22-PETQ-0007. 
For the purpose of open access, the authors have applied a CC-BY public copyright licence to any Author Accepted Manuscript (AAM) version arising from this submission.
\end{acknowledgments}


\bibliography{ICO_query_refs}

\begin{thebibliography}{48}%
\makeatletter
\providecommand \@ifxundefined [1]{%
 \@ifx{#1\undefined}
}%
\providecommand \@ifnum [1]{%
 \ifnum #1\expandafter \@firstoftwo
 \else \expandafter \@secondoftwo
 \fi
}%
\providecommand \@ifx [1]{%
 \ifx #1\expandafter \@firstoftwo
 \else \expandafter \@secondoftwo
 \fi
}%
\providecommand \natexlab [1]{#1}%
\providecommand \enquote  [1]{``#1''}%
\providecommand \bibnamefont  [1]{#1}%
\providecommand \bibfnamefont [1]{#1}%
\providecommand \citenamefont [1]{#1}%
\providecommand \href@noop [0]{\@secondoftwo}%
\providecommand \href [0]{\begingroup \@sanitize@url \@href}%
\providecommand \@href[1]{\@@startlink{#1}\@@href}%
\providecommand \@@href[1]{\endgroup#1\@@endlink}%
\providecommand \@sanitize@url [0]{\catcode `\\12\catcode `\$12\catcode
  `\&12\catcode `\#12\catcode `\^12\catcode `\_12\catcode `\%12\relax}%
\providecommand \@@startlink[1]{}%
\providecommand \@@endlink[0]{}%
\providecommand \url  [0]{\begingroup\@sanitize@url \@url }%
\providecommand \@url [1]{\endgroup\@href {#1}{\urlprefix }}%
\providecommand \urlprefix  [0]{URL }%
\providecommand \Eprint [0]{\href }%
\providecommand \doibase [0]{https://doi.org/}%
\providecommand \selectlanguage [0]{\@gobble}%
\providecommand \bibinfo  [0]{\@secondoftwo}%
\providecommand \bibfield  [0]{\@secondoftwo}%
\providecommand \translation [1]{[#1]}%
\providecommand \BibitemOpen [0]{}%
\providecommand \bibitemStop [0]{}%
\providecommand \bibitemNoStop [0]{.\EOS\space}%
\providecommand \EOS [0]{\spacefactor3000\relax}%
\providecommand \BibitemShut  [1]{\csname bibitem#1\endcsname}%
\let\auto@bib@innerbib\@empty
\bibitem [{\citenamefont {Bennett}\ \emph {et~al.}(1997)\citenamefont
  {Bennett}, \citenamefont {Bernstein}, \citenamefont {Brassard},\ and\
  \citenamefont {Vazirani}}]{bennett_strengths_1997}%
  \BibitemOpen
  \bibfield  {author} {\bibinfo {author} {\bibfnamefont {C.~H.}\ \bibnamefont
  {Bennett}}, \bibinfo {author} {\bibfnamefont {E.}~\bibnamefont {Bernstein}},
  \bibinfo {author} {\bibfnamefont {G.}~\bibnamefont {Brassard}},\ and\
  \bibinfo {author} {\bibfnamefont {U.}~\bibnamefont {Vazirani}},\ }\bibfield
  {title} {\bibinfo {title} {Strengths and {Weaknesses} of {Quantum}
  {Computing}},\ }\href {https://doi.org/10.1137/S0097539796300933} {\bibfield
  {journal} {\bibinfo  {journal} {SIAM Journal on Computing}\ }\textbf
  {\bibinfo {volume} {26}},\ \bibinfo {pages} {1510} (\bibinfo {year}
  {1997})},\ \Eprint {https://arxiv.org/abs/quant-ph/9701001}
  {arXiv:quant-ph/9701001 [quant-ph]} \BibitemShut {NoStop}%
\bibitem [{\citenamefont {Grover}(1996)}]{grover_fast_1996}%
  \BibitemOpen
  \bibfield  {author} {\bibinfo {author} {\bibfnamefont {L.~K.}\ \bibnamefont
  {Grover}},\ }\bibfield  {title} {\bibinfo {title} {A fast quantum mechanical
  algorithm for database search},\ }in\ \href
  {https://doi.org/10.1145/237814.237866} {\emph {\bibinfo {booktitle}
  {Proceedings of the Twenty-Eighth Annual {{ACM}} Symposium on {{Theory}} of
  Computing - {{STOC}} '96}}}\ (\bibinfo  {publisher} {{ACM Press}},\ \bibinfo
  {address} {{Philadelphia, Pennsylvania, United States}},\ \bibinfo {year}
  {1996})\ pp.\ \bibinfo {pages} {212--219},\ \Eprint
  {https://arxiv.org/abs/quant-ph/9605043} {arXiv:quant-ph/9605043 [quant-ph]}
  \BibitemShut {NoStop}%
\bibitem [{\citenamefont {D\"{u}rr}\ \emph {et~al.}(2006)\citenamefont
  {D\"{u}rr}, \citenamefont {Heiligman}, \citenamefont {HOyer},\ and\
  \citenamefont {Mhalla}}]{Durr_2006}%
  \BibitemOpen
  \bibfield  {author} {\bibinfo {author} {\bibfnamefont {C.}~\bibnamefont
  {D\"{u}rr}}, \bibinfo {author} {\bibfnamefont {M.}~\bibnamefont {Heiligman}},
  \bibinfo {author} {\bibfnamefont {P.}~\bibnamefont {HOyer}},\ and\ \bibinfo
  {author} {\bibfnamefont {M.}~\bibnamefont {Mhalla}},\ }\bibfield  {title}
  {\bibinfo {title} {Quantum query complexity of some graph problems},\ }\href
  {https://doi.org/10.1137/050644719} {\bibfield  {journal} {\bibinfo
  {journal} {SIAM Journal on Computing}\ }\textbf {\bibinfo {volume} {35}},\
  \bibinfo {pages} {1310} (\bibinfo {year} {2006})},\ \Eprint
  {https://arxiv.org/abs/quant-ph/0401091} {arXiv:quant-ph/0401091 [quant-ph]}
  \BibitemShut {NoStop}%
\bibitem [{\citenamefont {Chiribella}\ \emph {et~al.}(2013)\citenamefont
  {Chiribella}, \citenamefont {D'Ariano}, \citenamefont {Perinotti},\ and\
  \citenamefont {Valiron}}]{chiribella_quantum_2013}%
  \BibitemOpen
  \bibfield  {author} {\bibinfo {author} {\bibfnamefont {G.}~\bibnamefont
  {Chiribella}}, \bibinfo {author} {\bibfnamefont {G.~M.}\ \bibnamefont
  {D'Ariano}}, \bibinfo {author} {\bibfnamefont {P.}~\bibnamefont
  {Perinotti}},\ and\ \bibinfo {author} {\bibfnamefont {B.}~\bibnamefont
  {Valiron}},\ }\bibfield  {title} {\bibinfo {title} {Quantum computations
  without definite causal structure},\ }\href
  {https://doi.org/10.1103/PhysRevA.88.022318} {\bibfield  {journal} {\bibinfo
  {journal} {Phys. Rev. A}\ }\textbf {\bibinfo {volume} {88}},\ \bibinfo
  {pages} {022318} (\bibinfo {year} {2013})},\ \Eprint
  {https://arxiv.org/abs/0912.0195} {arXiv:0912.0195 [quant-ph]} \BibitemShut
  {NoStop}%
\bibitem [{\citenamefont {Oreshkov}\ \emph {et~al.}(2012)\citenamefont
  {Oreshkov}, \citenamefont {Costa},\ and\ \citenamefont
  {Brukner}}]{Oreshkov2012}%
  \BibitemOpen
  \bibfield  {author} {\bibinfo {author} {\bibfnamefont {O.}~\bibnamefont
  {Oreshkov}}, \bibinfo {author} {\bibfnamefont {F.}~\bibnamefont {Costa}},\
  and\ \bibinfo {author} {\bibfnamefont {{\v C}.}~\bibnamefont {Brukner}},\
  }\bibfield  {title} {\bibinfo {title} {Quantum correlations with no causal
  order},\ }\href {https://doi.org/10.1038/ncomms2076} {\bibfield  {journal}
  {\bibinfo  {journal} {Nat. Commun.}\ }\textbf {\bibinfo {volume} {3}},\
  \bibinfo {pages} {1092} (\bibinfo {year} {2012})},\ \Eprint
  {https://arxiv.org/abs/1105.4464} {arXiv:1105.4464 [quant-ph]} \BibitemShut
  {NoStop}%
\bibitem [{\citenamefont {Chiribella}\ \emph
  {et~al.}(2008{\natexlab{a}})\citenamefont {Chiribella}, \citenamefont
  {D'Ariano},\ and\ \citenamefont {Perinotti}}]{chiribella_quantum_2008}%
  \BibitemOpen
  \bibfield  {author} {\bibinfo {author} {\bibfnamefont {G.}~\bibnamefont
  {Chiribella}}, \bibinfo {author} {\bibfnamefont {G.~M.}\ \bibnamefont
  {D'Ariano}},\ and\ \bibinfo {author} {\bibfnamefont {P.}~\bibnamefont
  {Perinotti}},\ }\bibfield  {title} {\bibinfo {title} {Quantum circuit
  architecture},\ }\href {https://doi.org/10.1103/PhysRevLett.101.060401}
  {\bibfield  {journal} {\bibinfo  {journal} {Phys. Rev. Lett.}\ }\textbf
  {\bibinfo {volume} {101}},\ \bibinfo {pages} {060401} (\bibinfo {year}
  {2008}{\natexlab{a}})},\ \Eprint {https://arxiv.org/abs/0712.1325}
  {arXiv:0712.1325 [quant-ph]} \BibitemShut {NoStop}%
\bibitem [{\citenamefont {Oreshkov}\ and\ \citenamefont
  {Giarmatzi}(2016)}]{oreshkov_causal_2016}%
  \BibitemOpen
  \bibfield  {author} {\bibinfo {author} {\bibfnamefont {O.}~\bibnamefont
  {Oreshkov}}\ and\ \bibinfo {author} {\bibfnamefont {C.}~\bibnamefont
  {Giarmatzi}},\ }\bibfield  {title} {\bibinfo {title} {Causal and causally
  separable processes},\ }\href {https://doi.org/10.1088/1367-2630/18/9/093020}
  {\bibfield  {journal} {\bibinfo  {journal} {New Journal of Physics}\ }\textbf
  {\bibinfo {volume} {18}},\ \bibinfo {pages} {093020} (\bibinfo {year}
  {2016})},\ \Eprint {https://arxiv.org/abs/1506.05449} {arXiv:1506.05449
  [quant-ph]} \BibitemShut {NoStop}%
\bibitem [{\citenamefont {Wechs}\ \emph {et~al.}(2019)\citenamefont {Wechs},
  \citenamefont {Abbott},\ and\ \citenamefont {Branciard}}]{Wechs_2019}%
  \BibitemOpen
  \bibfield  {author} {\bibinfo {author} {\bibfnamefont {J.}~\bibnamefont
  {Wechs}}, \bibinfo {author} {\bibfnamefont {A.~A.}\ \bibnamefont {Abbott}},\
  and\ \bibinfo {author} {\bibfnamefont {C.}~\bibnamefont {Branciard}},\
  }\bibfield  {title} {\bibinfo {title} {On the definition and characterisation
  of multipartite causal (non)separability},\ }\href
  {https://doi.org/10.1088/1367-2630/aaf352} {\bibfield  {journal} {\bibinfo
  {journal} {New Journal of Physics}\ }\textbf {\bibinfo {volume} {21}},\
  \bibinfo {pages} {013027} (\bibinfo {year} {2019})},\ \Eprint
  {https://arxiv.org/abs/1807.10557} {arXiv:1807.10557 [quant-ph]} \BibitemShut
  {NoStop}%
\bibitem [{\citenamefont {Chiribella}(2012)}]{PhysRevA.86.040301}%
  \BibitemOpen
  \bibfield  {author} {\bibinfo {author} {\bibfnamefont {G.}~\bibnamefont
  {Chiribella}},\ }\bibfield  {title} {\bibinfo {title} {Perfect discrimination
  of no-signalling channels via quantum superposition of causal structures},\
  }\href {https://doi.org/10.1103/PhysRevA.86.040301} {\bibfield  {journal}
  {\bibinfo  {journal} {Phys. Rev. A}\ }\textbf {\bibinfo {volume} {86}},\
  \bibinfo {pages} {040301} (\bibinfo {year} {2012})},\ \Eprint
  {https://arxiv.org/abs/1109.5154} {arXiv:1109.5154 [quant-ph]} \BibitemShut
  {NoStop}%
\bibitem [{\citenamefont {Ara{\'u}jo}\ \emph {et~al.}(2014)\citenamefont
  {Ara{\'u}jo}, \citenamefont {Costa},\ and\ \citenamefont
  {Brukner}}]{PhysRevLett.113.250402}%
  \BibitemOpen
  \bibfield  {author} {\bibinfo {author} {\bibfnamefont {M.}~\bibnamefont
  {Ara{\'u}jo}}, \bibinfo {author} {\bibfnamefont {F.}~\bibnamefont {Costa}},\
  and\ \bibinfo {author} {\bibfnamefont {{\v C}.}~\bibnamefont {Brukner}},\
  }\bibfield  {title} {\bibinfo {title} {Computational advantage from
  quantum-controlled ordering of gates},\ }\href
  {https://doi.org/10.1103/PhysRevLett.113.250402} {\bibfield  {journal}
  {\bibinfo  {journal} {Phys. Rev. Lett.}\ }\textbf {\bibinfo {volume} {113}},\
  \bibinfo {pages} {250402} (\bibinfo {year} {2014})},\ \Eprint
  {https://arxiv.org/abs/1401.8127} {arXiv:1401.8127 [quant-ph]} \BibitemShut
  {NoStop}%
\bibitem [{\citenamefont {Renner}\ and\ \citenamefont
  {Brukner}(2022)}]{PhysRevLett.128.230503}%
  \BibitemOpen
  \bibfield  {author} {\bibinfo {author} {\bibfnamefont {M.~J.}\ \bibnamefont
  {Renner}}\ and\ \bibinfo {author} {\bibfnamefont {{\v C}.}~\bibnamefont
  {Brukner}},\ }\bibfield  {title} {\bibinfo {title} {Computational advantage
  from a quantum superposition of qubit gate orders},\ }\href
  {https://doi.org/10.1103/PhysRevLett.128.230503} {\bibfield  {journal}
  {\bibinfo  {journal} {Phys. Rev. Lett.}\ }\textbf {\bibinfo {volume} {128}},\
  \bibinfo {pages} {230503} (\bibinfo {year} {2022})},\ \Eprint
  {https://arxiv.org/abs/2112.14541} {arXiv:2112.14541 [quant-ph]} \BibitemShut
  {NoStop}%
\bibitem [{\citenamefont {Taddei}\ \emph {et~al.}(2021)\citenamefont {Taddei},
  \citenamefont {Cari\~ne}, \citenamefont {Mart\'{\i}nez}, \citenamefont
  {Garc\'{\i}a}, \citenamefont {Guerrero}, \citenamefont {Abbott},
  \citenamefont {Ara\'ujo}, \citenamefont {Branciard}, \citenamefont {G\'omez},
  \citenamefont {Walborn}, \citenamefont {Aolita},\ and\ \citenamefont
  {Lima}}]{PRXQuantum.2.010320}%
  \BibitemOpen
  \bibfield  {author} {\bibinfo {author} {\bibfnamefont {M.~M.}\ \bibnamefont
  {Taddei}}, \bibinfo {author} {\bibfnamefont {J.}~\bibnamefont {Cari\~ne}},
  \bibinfo {author} {\bibfnamefont {D.}~\bibnamefont {Mart\'{\i}nez}}, \bibinfo
  {author} {\bibfnamefont {T.}~\bibnamefont {Garc\'{\i}a}}, \bibinfo {author}
  {\bibfnamefont {N.}~\bibnamefont {Guerrero}}, \bibinfo {author}
  {\bibfnamefont {A.~A.}\ \bibnamefont {Abbott}}, \bibinfo {author}
  {\bibfnamefont {M.}~\bibnamefont {Ara\'ujo}}, \bibinfo {author}
  {\bibfnamefont {C.}~\bibnamefont {Branciard}}, \bibinfo {author}
  {\bibfnamefont {E.~S.}\ \bibnamefont {G\'omez}}, \bibinfo {author}
  {\bibfnamefont {S.~P.}\ \bibnamefont {Walborn}}, \bibinfo {author}
  {\bibfnamefont {L.}~\bibnamefont {Aolita}},\ and\ \bibinfo {author}
  {\bibfnamefont {G.}~\bibnamefont {Lima}},\ }\bibfield  {title} {\bibinfo
  {title} {Computational advantage from the quantum superposition of multiple
  temporal orders of photonic gates},\ }\href
  {https://doi.org/10.1103/PRXQuantum.2.010320} {\bibfield  {journal} {\bibinfo
   {journal} {PRX Quantum}\ }\textbf {\bibinfo {volume} {2}},\ \bibinfo {pages}
  {010320} (\bibinfo {year} {2021})},\ \Eprint
  {https://arxiv.org/abs/2002.07817} {arXiv:2002.07817 [quant-ph]} \BibitemShut
  {NoStop}%
\bibitem [{\citenamefont {Bavaresco}\ \emph {et~al.}(2022)\citenamefont
  {Bavaresco}, \citenamefont {Murao},\ and\ \citenamefont
  {Quintino}}]{bavaresco2022unitary}%
  \BibitemOpen
  \bibfield  {author} {\bibinfo {author} {\bibfnamefont {J.}~\bibnamefont
  {Bavaresco}}, \bibinfo {author} {\bibfnamefont {M.}~\bibnamefont {Murao}},\
  and\ \bibinfo {author} {\bibfnamefont {M.~T.}\ \bibnamefont {Quintino}},\
  }\bibfield  {title} {\bibinfo {title} {Unitary channel discrimination beyond
  group structures: Advantages of sequential and indefinite-causal-order
  strategies},\ }\href {https://doi.org/https://doi.org/10.1063/5.0075919}
  {\bibfield  {journal} {\bibinfo  {journal} {Journal of Mathematical Physics}\
  }\textbf {\bibinfo {volume} {63}},\ \bibinfo {pages} {042203} (\bibinfo
  {year} {2022})},\ \Eprint {https://arxiv.org/abs/2105.13369}
  {arXiv:2105.13369 [quant-ph]} \BibitemShut {NoStop}%
\bibitem [{\citenamefont {Quintino}\ and\ \citenamefont
  {Ebler}(2022)}]{quintino2022deterministic}%
  \BibitemOpen
  \bibfield  {author} {\bibinfo {author} {\bibfnamefont {M.~T.}\ \bibnamefont
  {Quintino}}\ and\ \bibinfo {author} {\bibfnamefont {D.}~\bibnamefont
  {Ebler}},\ }\bibfield  {title} {\bibinfo {title} {Deterministic
  transformations between unitary operations: {E}xponential advantage with
  adaptive quantum circuits and the power of indefinite causality},\ }\href
  {https://doi.org/10.22331/q-2022-03-31-679} {\bibfield  {journal} {\bibinfo
  {journal} {{Quantum}}\ }\textbf {\bibinfo {volume} {6}},\ \bibinfo {pages}
  {679} (\bibinfo {year} {2022})},\ \Eprint {https://arxiv.org/abs/2109.08202}
  {arXiv:2109.08202 [quant-ph]} \BibitemShut {NoStop}%
\bibitem [{\citenamefont {Zhao}\ \emph
  {et~al.}(2020{\natexlab{a}})\citenamefont {Zhao}, \citenamefont {Yang},\ and\
  \citenamefont {Chiribella}}]{PhysRevLett.124.190503}%
  \BibitemOpen
  \bibfield  {author} {\bibinfo {author} {\bibfnamefont {X.}~\bibnamefont
  {Zhao}}, \bibinfo {author} {\bibfnamefont {Y.}~\bibnamefont {Yang}},\ and\
  \bibinfo {author} {\bibfnamefont {G.}~\bibnamefont {Chiribella}},\ }\bibfield
   {title} {\bibinfo {title} {Quantum metrology with indefinite causal order},\
  }\href {https://doi.org/10.1103/PhysRevLett.124.190503} {\bibfield  {journal}
  {\bibinfo  {journal} {Phys. Rev. Lett.}\ }\textbf {\bibinfo {volume} {124}},\
  \bibinfo {pages} {190503} (\bibinfo {year} {2020}{\natexlab{a}})},\ \Eprint
  {https://arxiv.org/abs/1912.02449} {arXiv:1912.02449 [quant-ph]} \BibitemShut
  {NoStop}%
\bibitem [{\citenamefont {Ara\'ujo}\ \emph {et~al.}(2017)\citenamefont
  {Ara\'ujo}, \citenamefont {Gu\'erin},\ and\ \citenamefont
  {Baumeler}}]{PhysRevA.96.052315}%
  \BibitemOpen
  \bibfield  {author} {\bibinfo {author} {\bibfnamefont {M.}~\bibnamefont
  {Ara\'ujo}}, \bibinfo {author} {\bibfnamefont {P.~A.}\ \bibnamefont
  {Gu\'erin}},\ and\ \bibinfo {author} {\bibfnamefont {A.}~\bibnamefont
  {Baumeler}},\ }\bibfield  {title} {\bibinfo {title} {Quantum computation with
  indefinite causal structures},\ }\href
  {https://doi.org/10.1103/PhysRevA.96.052315} {\bibfield  {journal} {\bibinfo
  {journal} {Phys. Rev. A}\ }\textbf {\bibinfo {volume} {96}},\ \bibinfo
  {pages} {052315} (\bibinfo {year} {2017})},\ \Eprint
  {https://arxiv.org/abs/1706.09854} {arXiv:1706.09854 [quant-ph]} \BibitemShut
  {NoStop}%
\bibitem [{\citenamefont {Wechs}\ \emph {et~al.}(2021)\citenamefont {Wechs},
  \citenamefont {Dourdent}, \citenamefont {Abbott},\ and\ \citenamefont
  {Branciard}}]{wechs_quantum_2021}%
  \BibitemOpen
  \bibfield  {author} {\bibinfo {author} {\bibfnamefont {J.}~\bibnamefont
  {Wechs}}, \bibinfo {author} {\bibfnamefont {H.}~\bibnamefont {Dourdent}},
  \bibinfo {author} {\bibfnamefont {A.~A.}\ \bibnamefont {Abbott}},\ and\
  \bibinfo {author} {\bibfnamefont {C.}~\bibnamefont {Branciard}},\ }\bibfield
  {title} {\bibinfo {title} {Quantum {Circuits} with {Classical} {Versus}
  {Quantum} {Control} of {Causal} {Order}},\ }\href
  {https://doi.org/10.1103/PRXQuantum.2.030335} {\bibfield  {journal} {\bibinfo
   {journal} {PRX Quantum}\ }\textbf {\bibinfo {volume} {2}},\ \bibinfo {pages}
  {030335} (\bibinfo {year} {2021})},\ \Eprint
  {https://arxiv.org/abs/2101.08796} {arXiv:2101.08796 [quant-ph]} \BibitemShut
  {NoStop}%
\bibitem [{\citenamefont {Buhrman}\ and\ \citenamefont {{de
  Wolf}}(2002)}]{BUHRMAN200221}%
  \BibitemOpen
  \bibfield  {author} {\bibinfo {author} {\bibfnamefont {H.}~\bibnamefont
  {Buhrman}}\ and\ \bibinfo {author} {\bibfnamefont {R.}~\bibnamefont {{de
  Wolf}}},\ }\bibfield  {title} {\bibinfo {title} {Complexity measures and
  decision tree complexity: a survey},\ }\href
  {https://doi.org/https://doi.org/10.1016/S0304-3975(01)00144-X} {\bibfield
  {journal} {\bibinfo  {journal} {Theoretical Computer Science}\ }\textbf
  {\bibinfo {volume} {288}},\ \bibinfo {pages} {21} (\bibinfo {year}
  {2002})}\BibitemShut {NoStop}%
\bibitem [{\citenamefont {Ambainis}(2018)}]{ambainis2018understanding}%
  \BibitemOpen
  \bibfield  {author} {\bibinfo {author} {\bibfnamefont {A.}~\bibnamefont
  {Ambainis}},\ }\bibfield  {title} {\bibinfo {title} {Understanding quantum
  algorithms via query complexity},\ }in\ \href
  {https://doi.org/https://doi.org/10.1142/9789813272880_0181} {\emph {\bibinfo
  {booktitle} {Proceedings of the International Congress of Mathematicians: Rio
  de Janeiro 2018}}}\ (\bibinfo {organization} {World Scientific},\ \bibinfo
  {year} {2018})\ pp.\ \bibinfo {pages} {3265--3285},\ \Eprint
  {https://arxiv.org/abs/1712.06349} {arXiv:1712.06349 [quant-ph]} \BibitemShut
  {NoStop}%
\bibitem [{\citenamefont {Ambainis}(2002)}]{ambainis_quantum_2002}%
  \BibitemOpen
  \bibfield  {author} {\bibinfo {author} {\bibfnamefont {A.}~\bibnamefont
  {Ambainis}},\ }\bibfield  {title} {\bibinfo {title} {Quantum {Lower} {Bounds}
  by {Quantum} {Arguments}},\ }\href {https://doi.org/10.1006/jcss.2002.1826}
  {\bibfield  {journal} {\bibinfo  {journal} {Journal of Computer and System
  Sciences}\ }\textbf {\bibinfo {volume} {64}},\ \bibinfo {pages} {750}
  (\bibinfo {year} {2002})},\ \Eprint {https://arxiv.org/abs/quant-ph/0002066}
  {arXiv:quant-ph/0002066 [quant-ph]} \BibitemShut {NoStop}%
\bibitem [{\citenamefont {Choi}(1975)}]{CHOI1975285}%
  \BibitemOpen
  \bibfield  {author} {\bibinfo {author} {\bibfnamefont {M.-D.}\ \bibnamefont
  {Choi}},\ }\bibfield  {title} {\bibinfo {title} {Completely positive linear
  maps on complex matrices},\ }\href
  {https://doi.org/10.1016/0024-3795(75)90075-0} {\bibfield  {journal}
  {\bibinfo  {journal} {Linear Algebra and its Applications}\ }\textbf
  {\bibinfo {volume} {10}},\ \bibinfo {pages} {285} (\bibinfo {year}
  {1975})}\BibitemShut {NoStop}%
\bibitem [{\citenamefont {Jamio{\l}kowski}(1972)}]{JAMIOLKOWSKI1972275}%
  \BibitemOpen
  \bibfield  {author} {\bibinfo {author} {\bibfnamefont {A.}~\bibnamefont
  {Jamio{\l}kowski}},\ }\bibfield  {title} {\bibinfo {title} {Linear
  transformations which preserve trace and positive semidefiniteness of
  operators},\ }\href
  {https://doi.org/https://doi.org/10.1016/0034-4877(72)90011-0} {\bibfield
  {journal} {\bibinfo  {journal} {Reports on Mathematical Physics}\ }\textbf
  {\bibinfo {volume} {3}},\ \bibinfo {pages} {275} (\bibinfo {year}
  {1972})}\BibitemShut {NoStop}%
\bibitem [{\citenamefont {Chiribella}\ \emph {et~al.}(2009)\citenamefont
  {Chiribella}, \citenamefont {D'Ariano},\ and\ \citenamefont
  {Perinotti}}]{chiribella_theoretical_2009}%
  \BibitemOpen
  \bibfield  {author} {\bibinfo {author} {\bibfnamefont {G.}~\bibnamefont
  {Chiribella}}, \bibinfo {author} {\bibfnamefont {G.~M.}\ \bibnamefont
  {D'Ariano}},\ and\ \bibinfo {author} {\bibfnamefont {P.}~\bibnamefont
  {Perinotti}},\ }\bibfield  {title} {\bibinfo {title} {Theoretical framework
  for quantum networks},\ }\href {https://doi.org/10.1103/PhysRevA.80.022339}
  {\bibfield  {journal} {\bibinfo  {journal} {Phys. Rev. A}\ }\textbf {\bibinfo
  {volume} {80}},\ \bibinfo {pages} {022339} (\bibinfo {year} {2009})},\
  \Eprint {https://arxiv.org/abs/0904.4483} {arXiv:0904.4483 [quant-ph]}
  \BibitemShut {NoStop}%
\bibitem [{\citenamefont {Chiribella}\ \emph
  {et~al.}(2008{\natexlab{b}})\citenamefont {Chiribella}, \citenamefont
  {D'Ariano},\ and\ \citenamefont {Perinotti}}]{chiribella_transforming_2008}%
  \BibitemOpen
  \bibfield  {author} {\bibinfo {author} {\bibfnamefont {G.}~\bibnamefont
  {Chiribella}}, \bibinfo {author} {\bibfnamefont {G.~M.}\ \bibnamefont
  {D'Ariano}},\ and\ \bibinfo {author} {\bibfnamefont {P.}~\bibnamefont
  {Perinotti}},\ }\bibfield  {title} {\bibinfo {title} {Transforming quantum
  operations: {Quantum} supermaps},\ }\href
  {https://doi.org/10.1209/0295-5075/83/30004} {\bibfield  {journal} {\bibinfo
  {journal} {EPL (Europhysics Letters)}\ }\textbf {\bibinfo {volume} {83}},\
  \bibinfo {pages} {30004} (\bibinfo {year} {2008}{\natexlab{b}})},\ \Eprint
  {https://arxiv.org/abs/0804.0180} {arXiv:0804.0180 [quant-ph]} \BibitemShut
  {NoStop}%
\bibitem [{\citenamefont {Ara{\'u}jo}\ \emph {et~al.}(2015)\citenamefont
  {Ara{\'u}jo}, \citenamefont {Branciard}, \citenamefont {Costa}, \citenamefont
  {Feix}, \citenamefont {Giarmatzi},\ and\ \citenamefont
  {Brukner}}]{araujo_witnessing_2015}%
  \BibitemOpen
  \bibfield  {author} {\bibinfo {author} {\bibfnamefont {M.}~\bibnamefont
  {Ara{\'u}jo}}, \bibinfo {author} {\bibfnamefont {C.}~\bibnamefont
  {Branciard}}, \bibinfo {author} {\bibfnamefont {F.}~\bibnamefont {Costa}},
  \bibinfo {author} {\bibfnamefont {A.}~\bibnamefont {Feix}}, \bibinfo {author}
  {\bibfnamefont {C.}~\bibnamefont {Giarmatzi}},\ and\ \bibinfo {author}
  {\bibfnamefont {{\v C}.}~\bibnamefont {Brukner}},\ }\bibfield  {title}
  {\bibinfo {title} {Witnessing causal nonseparability},\ }\href
  {https://doi.org/10.1088/1367-2630/17/10/102001} {\bibfield  {journal}
  {\bibinfo  {journal} {New Journal of Physics}\ }\textbf {\bibinfo {volume}
  {17}},\ \bibinfo {pages} {102001} (\bibinfo {year} {2015})},\ \Eprint
  {https://arxiv.org/abs/1506.03776} {arXiv:1506.03776 [quant-ph]} \BibitemShut
  {NoStop}%
\bibitem [{\citenamefont {Purves}\ and\ \citenamefont
  {Short}(2021)}]{purves21}%
  \BibitemOpen
  \bibfield  {author} {\bibinfo {author} {\bibfnamefont {T.}~\bibnamefont
  {Purves}}\ and\ \bibinfo {author} {\bibfnamefont {A.~J.}\ \bibnamefont
  {Short}},\ }\bibfield  {title} {\bibinfo {title} {Quantum theory cannot
  violate a causal inequality},\ }\href
  {https://doi.org/10.1103/PhysRevLett.127.110402} {\bibfield  {journal}
  {\bibinfo  {journal} {Phys. Rev. Lett.}\ }\textbf {\bibinfo {volume} {127}},\
  \bibinfo {pages} {110402} (\bibinfo {year} {2021})},\ \Eprint
  {https://arxiv.org/abs/2101.09107} {arXiv:2101.09107 [quant-ph]} \BibitemShut
  {NoStop}%
\bibitem [{\citenamefont {Salzger}(2023)}]{salzger23}%
  \BibitemOpen
  \bibfield  {author} {\bibinfo {author} {\bibfnamefont {M.}~\bibnamefont
  {Salzger}},\ }\bibfield  {title} {\bibinfo {title} {Connecting indefinite
  causal order processes to composable quantum protocols in a spacetime},\
  }\Eprint {https://arxiv.org/abs/2304.06735} {arXiv:2304.06735 [quant-ph]}
  (\bibinfo {year} {2023})\BibitemShut {NoStop}%
\bibitem [{\citenamefont {Quintino}\ \emph {et~al.}(2019)\citenamefont
  {Quintino}, \citenamefont {Dong}, \citenamefont {Shimbo}, \citenamefont
  {Soeda},\ and\ \citenamefont {Murao}}]{quintino19}%
  \BibitemOpen
  \bibfield  {author} {\bibinfo {author} {\bibfnamefont {M.~T.}\ \bibnamefont
  {Quintino}}, \bibinfo {author} {\bibfnamefont {Q.}~\bibnamefont {Dong}},
  \bibinfo {author} {\bibfnamefont {A.}~\bibnamefont {Shimbo}}, \bibinfo
  {author} {\bibfnamefont {A.}~\bibnamefont {Soeda}},\ and\ \bibinfo {author}
  {\bibfnamefont {M.}~\bibnamefont {Murao}},\ }\bibfield  {title} {\bibinfo
  {title} {Reversing unknown quantum transformations: Universal quantum circuit
  for inverting general unitary operations},\ }\href
  {https://doi.org/10.1103/PhysRevLett.123.210502} {\bibfield  {journal}
  {\bibinfo  {journal} {Phys. Rev. Lett.}\ }\textbf {\bibinfo {volume} {123}},\
  \bibinfo {pages} {210502} (\bibinfo {year} {2019})},\ \Eprint
  {https://arxiv.org/abs/1810.06944} {arXiv:1810.06944 [quant-ph]} \BibitemShut
  {NoStop}%
\bibitem [{\citenamefont {Bavaresco}\ \emph {et~al.}(2021)\citenamefont
  {Bavaresco}, \citenamefont {Murao},\ and\ \citenamefont
  {Quintino}}]{PhysRevLett.127.200504}%
  \BibitemOpen
  \bibfield  {author} {\bibinfo {author} {\bibfnamefont {J.}~\bibnamefont
  {Bavaresco}}, \bibinfo {author} {\bibfnamefont {M.}~\bibnamefont {Murao}},\
  and\ \bibinfo {author} {\bibfnamefont {M.~T.}\ \bibnamefont {Quintino}},\
  }\bibfield  {title} {\bibinfo {title} {Strict hierarchy between parallel,
  sequential, and indefinite-causal-order strategies for channel
  discrimination},\ }\href {https://doi.org/10.1103/PhysRevLett.127.200504}
  {\bibfield  {journal} {\bibinfo  {journal} {Phys. Rev. Lett.}\ }\textbf
  {\bibinfo {volume} {127}},\ \bibinfo {pages} {200504} (\bibinfo {year}
  {2021})},\ \Eprint {https://arxiv.org/abs/2011.08300} {arXiv:2011.08300
  [quant-ph]} \BibitemShut {NoStop}%
\bibitem [{\citenamefont {Liu}\ \emph {et~al.}(2023)\citenamefont {Liu},
  \citenamefont {Hu}, \citenamefont {Yuan},\ and\ \citenamefont
  {Yang}}]{liu23}%
  \BibitemOpen
  \bibfield  {author} {\bibinfo {author} {\bibfnamefont {Q.}~\bibnamefont
  {Liu}}, \bibinfo {author} {\bibfnamefont {Z.}~\bibnamefont {Hu}}, \bibinfo
  {author} {\bibfnamefont {H.}~\bibnamefont {Yuan}},\ and\ \bibinfo {author}
  {\bibfnamefont {Y.}~\bibnamefont {Yang}},\ }\bibfield  {title} {\bibinfo
  {title} {Optimal strategies of quantum metrology with a strict hierarchy},\
  }\href {https://doi.org/10.1103/PhysRevLett.130.070803} {\bibfield  {journal}
  {\bibinfo  {journal} {Phys. Rev. Lett.}\ }\textbf {\bibinfo {volume} {130}},\
  \bibinfo {pages} {070803} (\bibinfo {year} {2023})},\ \Eprint
  {https://arxiv.org/abs/2203.09758} {arXiv:2203.09758} \BibitemShut {NoStop}%
\bibitem [{\citenamefont {Mothe}\ \emph {et~al.}(2023)\citenamefont {Mothe},
  \citenamefont {Branciard},\ and\ \citenamefont {Abbott}}]{mothe23}%
  \BibitemOpen
  \bibfield  {author} {\bibinfo {author} {\bibfnamefont {R.}~\bibnamefont
  {Mothe}}, \bibinfo {author} {\bibfnamefont {C.}~\bibnamefont {Branciard}},\
  and\ \bibinfo {author} {\bibfnamefont {A.~A.}\ \bibnamefont {Abbott}},\
  }\bibfield  {title} {\bibinfo {title} {Reassessing the advantage of
  indefinite causal orders for quantum metrology},\ }\Eprint
  {https://arxiv.org/abs/2312.12172} {arXiv:2312.12172 [quant-ph]}  (\bibinfo
  {year} {2023})\BibitemShut {NoStop}%
\bibitem [{\citenamefont {Beals}\ \emph {et~al.}(2001)\citenamefont {Beals},
  \citenamefont {Buhrman}, \citenamefont {Cleve}, \citenamefont {Mosca},\ and\
  \citenamefont {de~Wolf}}]{beals_quantum_2001}%
  \BibitemOpen
  \bibfield  {author} {\bibinfo {author} {\bibfnamefont {R.}~\bibnamefont
  {Beals}}, \bibinfo {author} {\bibfnamefont {H.}~\bibnamefont {Buhrman}},
  \bibinfo {author} {\bibfnamefont {R.}~\bibnamefont {Cleve}}, \bibinfo
  {author} {\bibfnamefont {M.}~\bibnamefont {Mosca}},\ and\ \bibinfo {author}
  {\bibfnamefont {R.}~\bibnamefont {de~Wolf}},\ }\bibfield  {title} {\bibinfo
  {title} {Quantum lower bounds by polynomials},\ }\href
  {https://doi.org/10.1145/502090.502097} {\bibfield  {journal} {\bibinfo
  {journal} {Journal of the ACM}\ }\textbf {\bibinfo {volume} {48}},\ \bibinfo
  {pages} {778} (\bibinfo {year} {2001})},\ \Eprint
  {https://arxiv.org/abs/quant-ph/9802049} {arXiv:quant-ph/9802049 [quant-ph]}
  \BibitemShut {NoStop}%
\bibitem [{\citenamefont {Ambainis}(2006)}]{AMBAINIS2006220}%
  \BibitemOpen
  \bibfield  {author} {\bibinfo {author} {\bibfnamefont {A.}~\bibnamefont
  {Ambainis}},\ }\bibfield  {title} {\bibinfo {title} {Polynomial degree vs.
  quantum query complexity},\ }\href
  {https://doi.org/https://doi.org/10.1016/j.jcss.2005.06.006} {\bibfield
  {journal} {\bibinfo  {journal} {Journal of Computer and System Sciences}\
  }\textbf {\bibinfo {volume} {72}},\ \bibinfo {pages} {220} (\bibinfo {year}
  {2006})},\ \Eprint {https://arxiv.org/abs/quant-ph/0305028}
  {arXiv:quant-ph/0305028 [quant-ph]} \BibitemShut {NoStop}%
\bibitem [{\citenamefont {Aaronson}\ \emph {et~al.}(2016)\citenamefont
  {Aaronson}, \citenamefont {Ben-David},\ and\ \citenamefont
  {Kothari}}]{aaronson16}%
  \BibitemOpen
  \bibfield  {author} {\bibinfo {author} {\bibfnamefont {S.}~\bibnamefont
  {Aaronson}}, \bibinfo {author} {\bibfnamefont {S.}~\bibnamefont
  {Ben-David}},\ and\ \bibinfo {author} {\bibfnamefont {R.}~\bibnamefont
  {Kothari}},\ }\bibfield  {title} {\bibinfo {title} {Separations in query
  complexity using cheat sheets},\ }in\ \href
  {https://doi.org/10.1145/2897518.2897644} {\emph {\bibinfo {booktitle}
  {Proceedings of the 48th Annual ACM SIGACT Symposium on Theory of Computing
  (STOC 2016)}}}\ (\bibinfo {year} {2016})\ pp.\ \bibinfo {pages} {863--876},\
  \Eprint {https://arxiv.org/abs/1511.01937} {arXiv:1511.01937 [quant-ph]}
  \BibitemShut {NoStop}%
\bibitem [{\citenamefont {Aaronson}\ \emph {et~al.}(2021)\citenamefont
  {Aaronson}, \citenamefont {Ben-David}, \citenamefont {Kothari}, \citenamefont
  {Rao},\ and\ \citenamefont {Tal}}]{aaronson2021degree}%
  \BibitemOpen
  \bibfield  {author} {\bibinfo {author} {\bibfnamefont {S.}~\bibnamefont
  {Aaronson}}, \bibinfo {author} {\bibfnamefont {S.}~\bibnamefont {Ben-David}},
  \bibinfo {author} {\bibfnamefont {R.}~\bibnamefont {Kothari}}, \bibinfo
  {author} {\bibfnamefont {S.}~\bibnamefont {Rao}},\ and\ \bibinfo {author}
  {\bibfnamefont {A.}~\bibnamefont {Tal}},\ }\bibfield  {title} {\bibinfo
  {title} {Degree vs. approximate degree and quantum implications of huang's
  sensitivity theorem},\ }in\ \href
  {https://dl.acm.org/doi/10.1145/3406325.3451047} {\emph {\bibinfo {booktitle}
  {Proceedings of the 53rd Annual ACM SIGACT Symposium on Theory of
  Computing}}}\ (\bibinfo {year} {2021})\ pp.\ \bibinfo {pages} {1330--1342},\
  \Eprint {https://arxiv.org/abs/2010.12629} {arXiv:2010.12629 [quant-ph]}
  \BibitemShut {NoStop}%
\bibitem [{\citenamefont {Montanaro}\ \emph {et~al.}(2015)\citenamefont
  {Montanaro}, \citenamefont {Jozsa},\ and\ \citenamefont
  {Mitchison}}]{Montanaro2015}%
  \BibitemOpen
  \bibfield  {author} {\bibinfo {author} {\bibfnamefont {A.}~\bibnamefont
  {Montanaro}}, \bibinfo {author} {\bibfnamefont {R.}~\bibnamefont {Jozsa}},\
  and\ \bibinfo {author} {\bibfnamefont {G.}~\bibnamefont {Mitchison}},\
  }\bibfield  {title} {\bibinfo {title} {On exact quantum query complexity},\
  }\href {https://doi.org/10.1007/s00453-013-9826-8} {\bibfield  {journal}
  {\bibinfo  {journal} {Algorithmica}\ }\textbf {\bibinfo {volume} {71}},\
  \bibinfo {pages} {775} (\bibinfo {year} {2015})},\ \Eprint
  {https://arxiv.org/abs/1111.0475} {arXiv:1111.0475 [quant-ph]} \BibitemShut
  {NoStop}%
\bibitem [{\citenamefont {Barnum}\ \emph {et~al.}(2003)\citenamefont {Barnum},
  \citenamefont {Saks},\ and\ \citenamefont {Szegedy}}]{QuerySDP}%
  \BibitemOpen
  \bibfield  {author} {\bibinfo {author} {\bibfnamefont {H.}~\bibnamefont
  {Barnum}}, \bibinfo {author} {\bibfnamefont {M.}~\bibnamefont {Saks}},\ and\
  \bibinfo {author} {\bibfnamefont {M.}~\bibnamefont {Szegedy}},\ }\bibfield
  {title} {\bibinfo {title} {Quantum query complexity and semi-definite
  programming},\ }in\ \href {https://doi.org/10.1109/CCC.2003.1214419} {\emph
  {\bibinfo {booktitle} {18th IEEE Annual Conference on Computational
  Complexity, 2003. Proceedings.}}}\ (\bibinfo {year} {2003})\ pp.\ \bibinfo
  {pages} {179--193}\BibitemShut {NoStop}%
\bibitem [{\citenamefont {Hoyer}\ \emph {et~al.}(2007)\citenamefont {Hoyer},
  \citenamefont {Lee},\ and\ \citenamefont
  {Spalek}}]{hoyerNegativeWeightsMake2007a}%
  \BibitemOpen
  \bibfield  {author} {\bibinfo {author} {\bibfnamefont {P.}~\bibnamefont
  {Hoyer}}, \bibinfo {author} {\bibfnamefont {T.}~\bibnamefont {Lee}},\ and\
  \bibinfo {author} {\bibfnamefont {R.}~\bibnamefont {Spalek}},\ }\bibfield
  {title} {\bibinfo {title} {Negative weights make adversaries stronger},\ }in\
  \href {https://doi.org/10.1145/1250790.1250867} {\emph {\bibinfo {booktitle}
  {Proceedings of the Thirty-Ninth Annual {{ACM}} Symposium on {{Theory}} of
  Computing - {{STOC}} '07}}}\ (\bibinfo  {publisher} {{ACM Press}},\ \bibinfo
  {address} {{San Diego, California, USA}},\ \bibinfo {year} {2007})\ p.\
  \bibinfo {pages} {526},\ \Eprint {https://arxiv.org/abs/quant-ph/0611054}
  {arXiv:quant-ph/0611054 [quant-ph]} \BibitemShut {NoStop}%
\bibitem [{\citenamefont {Gu{\'{e}}rin}\ \emph {et~al.}(2019)\citenamefont
  {Gu{\'{e}}rin}, \citenamefont {Krumm}, \citenamefont {Budroni},\ and\
  \citenamefont {Brukner}}]{guerin19}%
  \BibitemOpen
  \bibfield  {author} {\bibinfo {author} {\bibfnamefont {P.~A.}\ \bibnamefont
  {Gu{\'{e}}rin}}, \bibinfo {author} {\bibfnamefont {M.}~\bibnamefont {Krumm}},
  \bibinfo {author} {\bibfnamefont {C.}~\bibnamefont {Budroni}},\ and\ \bibinfo
  {author} {\bibfnamefont {{\v{C}}.}~\bibnamefont {Brukner}},\ }\bibfield
  {title} {\bibinfo {title} {Composition rules for quantum processes: a no-go
  theorem},\ }\href {https://doi.org/10.1088/1367-2630/aafef7} {\bibfield
  {journal} {\bibinfo  {journal} {New J. Phys.}\ }\textbf {\bibinfo {volume}
  {21}},\ \bibinfo {pages} {012001} (\bibinfo {year} {2019})},\ \Eprint
  {https://arxiv.org/abs/1806.10374} {arXiv:1806.10374 [quant-ph]} \BibitemShut
  {NoStop}%
\bibitem [{\citenamefont {Kissinger}\ and\ \citenamefont
  {Uijlen}(2019)}]{kissinger19}%
  \BibitemOpen
  \bibfield  {author} {\bibinfo {author} {\bibfnamefont {A.}~\bibnamefont
  {Kissinger}}\ and\ \bibinfo {author} {\bibfnamefont {S.}~\bibnamefont
  {Uijlen}},\ }\bibfield  {title} {\bibinfo {title} {A categorical semantics
  for causal structure},\ }\href {https://doi.org/10.23638/LMCS-15(3:15)2019}
  {\bibfield  {journal} {\bibinfo  {journal} {Logical Methods Comput. Sci.}\
  }\textbf {\bibinfo {volume} {15}},\ \bibinfo {pages} {4426} (\bibinfo {year}
  {2019})},\ \Eprint {https://arxiv.org/abs/1701.04732} {arXiv:1701.04732
  [quant-ph]} \BibitemShut {NoStop}%
\bibitem [{\citenamefont {Giovannetti}\ \emph {et~al.}(2011)\citenamefont
  {Giovannetti}, \citenamefont {Llyod},\ and\ \citenamefont
  {Maccone}}]{giovannetti11}%
  \BibitemOpen
  \bibfield  {author} {\bibinfo {author} {\bibfnamefont {V.}~\bibnamefont
  {Giovannetti}}, \bibinfo {author} {\bibfnamefont {S.}~\bibnamefont {Llyod}},\
  and\ \bibinfo {author} {\bibfnamefont {L.}~\bibnamefont {Maccone}},\
  }\bibfield  {title} {\bibinfo {title} {Advances in quantum metrology},\
  }\href {https://doi.org/10.1038/nphoton.2011.35} {\bibfield  {journal}
  {\bibinfo  {journal} {Nat. Photonics}\ }\textbf {\bibinfo {volume} {5}},\
  \bibinfo {pages} {222} (\bibinfo {year} {2011})},\ \Eprint
  {https://arxiv.org/abs/1102.2318} {arXiv:1102.2318 [quant-ph]} \BibitemShut
  {NoStop}%
\bibitem [{\citenamefont {Zhao}\ \emph
  {et~al.}(2020{\natexlab{b}})\citenamefont {Zhao}, \citenamefont {Yang},\ and\
  \citenamefont {Chiribella}}]{zhao}%
  \BibitemOpen
  \bibfield  {author} {\bibinfo {author} {\bibfnamefont {X.}~\bibnamefont
  {Zhao}}, \bibinfo {author} {\bibfnamefont {Y.}~\bibnamefont {Yang}},\ and\
  \bibinfo {author} {\bibfnamefont {G.}~\bibnamefont {Chiribella}},\ }\bibfield
   {title} {\bibinfo {title} {Quantum metrology with indefinite causal order},\
  }\href {https://doi.org/10.1103/PhysRevLett.124.190503} {\bibfield  {journal}
  {\bibinfo  {journal} {Phys. Rev. Lett.}\ }\textbf {\bibinfo {volume} {124}},\
  \bibinfo {pages} {190503} (\bibinfo {year} {2020}{\natexlab{b}})},\ \Eprint
  {https://arxiv.org/abs/1912.02449} {arXiv:1912.02449 [quant-ph]} \BibitemShut
  {NoStop}%
\bibitem [{\citenamefont {Ara{\'{u}}jo}\ \emph {et~al.}(2017)\citenamefont
  {Ara{\'{u}}jo}, \citenamefont {Feix}, \citenamefont {Navascu{\'{e}}s},\ and\
  \citenamefont {Brukner}}]{Araujo2017purification}%
  \BibitemOpen
  \bibfield  {author} {\bibinfo {author} {\bibfnamefont {M.}~\bibnamefont
  {Ara{\'{u}}jo}}, \bibinfo {author} {\bibfnamefont {A.}~\bibnamefont {Feix}},
  \bibinfo {author} {\bibfnamefont {M.}~\bibnamefont {Navascu{\'{e}}s}},\ and\
  \bibinfo {author} {\bibfnamefont {{\v C}.}~\bibnamefont {Brukner}},\
  }\bibfield  {title} {\bibinfo {title} {A purification postulate for quantum
  mechanics with indefinite causal order},\ }\href
  {https://doi.org/10.22331/q-2017-04-26-10} {\bibfield  {journal} {\bibinfo
  {journal} {{Quantum}}\ }\textbf {\bibinfo {volume} {1}},\ \bibinfo {pages}
  {10} (\bibinfo {year} {2017})},\ \Eprint {https://arxiv.org/abs/1611.08535}
  {arXiv:1611.08535 [quant-ph]} \BibitemShut {NoStop}%
\bibitem [{\citenamefont {Wechs}\ \emph {et~al.}(2023)\citenamefont {Wechs},
  \citenamefont {Branciard},\ and\ \citenamefont {Oreshkov}}]{wechs23}%
  \BibitemOpen
  \bibfield  {author} {\bibinfo {author} {\bibfnamefont {J.}~\bibnamefont
  {Wechs}}, \bibinfo {author} {\bibfnamefont {C.}~\bibnamefont {Branciard}},\
  and\ \bibinfo {author} {\bibfnamefont {O.}~\bibnamefont {Oreshkov}},\
  }\bibfield  {title} {\bibinfo {title} {Existence of processes violating
  causal inequalities on time-delocalised subsystems},\ }\href
  {https://doi.org/10.1038/s41467-023-36893-3} {\bibfield  {journal} {\bibinfo
  {journal} {Nat. Commun.}\ }\textbf {\bibinfo {volume} {14}},\ \bibinfo
  {pages} {1471} (\bibinfo {year} {2023})},\ \Eprint
  {https://arxiv.org/abs/2201.11832} {arXiv:2201.11832 [quant-ph]} \BibitemShut
  {NoStop}%
\bibitem [{\citenamefont {Gu{\'{e}}rin}\ \emph {et~al.}(2016)\citenamefont
  {Gu{\'{e}}rin}, \citenamefont {Feix}, \citenamefont {Ara{\'u}jo},\ and\
  \citenamefont {Brukner}}]{guerin16}%
  \BibitemOpen
  \bibfield  {author} {\bibinfo {author} {\bibfnamefont {P.~A.}\ \bibnamefont
  {Gu{\'{e}}rin}}, \bibinfo {author} {\bibfnamefont {A.}~\bibnamefont {Feix}},
  \bibinfo {author} {\bibfnamefont {M.}~\bibnamefont {Ara{\'u}jo}},\ and\
  \bibinfo {author} {\bibfnamefont {{\v{C}}.}~\bibnamefont {Brukner}},\
  }\bibfield  {title} {\bibinfo {title} {Exponential communication complexity
  advantage from quantum superposition of the direction of communication},\
  }\href {https://doi.org/10.1103/PhysRevLett.117.100502} {\bibfield  {journal}
  {\bibinfo  {journal} {Phys. Rev. Lett.}\ }\textbf {\bibinfo {volume} {117}},\
  \bibinfo {pages} {100502} (\bibinfo {year} {2016})},\ \Eprint
  {https://arxiv.org/abs/1605.07372} {arXiv:1605.07372 [quant-ph]} \BibitemShut
  {NoStop}%
\bibitem [{\citenamefont {Boyd}\ and\ \citenamefont
  {Vandenberghe}(2004)}]{boyd2004convex}%
  \BibitemOpen
  \bibfield  {author} {\bibinfo {author} {\bibfnamefont {S.}~\bibnamefont
  {Boyd}}\ and\ \bibinfo {author} {\bibfnamefont {L.}~\bibnamefont
  {Vandenberghe}},\ }\href {https://doi.org/10.1017/CBO9780511804441} {\emph
  {\bibinfo {title} {Convex Optimization}}}\ (\bibinfo  {publisher} {Cambridge
  University Press},\ \bibinfo {year} {2004})\BibitemShut {NoStop}%
\bibitem [{\citenamefont {Lofberg}(2004)}]{Lofberg2004}%
  \BibitemOpen
  \bibfield  {author} {\bibinfo {author} {\bibfnamefont {J.}~\bibnamefont
  {Lofberg}},\ }\bibfield  {title} {\bibinfo {title} {Yalmip : a toolbox for
  modeling and optimization in matlab},\ }in\ \href
  {https://doi.org/10.1109/CACSD.2004.1393890} {\emph {\bibinfo {booktitle}
  {2004 IEEE International Conference on Robotics and Automation (IEEE Cat.
  No.04CH37508)}}}\ (\bibinfo {year} {2004})\ pp.\ \bibinfo {pages}
  {284--289}\BibitemShut {NoStop}%
\bibitem [{\citenamefont {O'Donoghue}\ \emph {et~al.}(2016)\citenamefont
  {O'Donoghue}, \citenamefont {Chu}, \citenamefont {Parikh},\ and\
  \citenamefont {Boyd}}]{scs}%
  \BibitemOpen
  \bibfield  {author} {\bibinfo {author} {\bibfnamefont {B.}~\bibnamefont
  {O'Donoghue}}, \bibinfo {author} {\bibfnamefont {E.}~\bibnamefont {Chu}},
  \bibinfo {author} {\bibfnamefont {N.}~\bibnamefont {Parikh}},\ and\ \bibinfo
  {author} {\bibfnamefont {S.}~\bibnamefont {Boyd}},\ }\bibfield  {title}
  {\bibinfo {title} {Conic optimization via operator splitting and homogeneous
  self-dual embedding},\ }\href {http://stanford.edu/~boyd/papers/scs.html}
  {\bibfield  {journal} {\bibinfo  {journal} {Journal of Optimization Theory
  and Applications}\ }\textbf {\bibinfo {volume} {169}},\ \bibinfo {pages}
  {1042} (\bibinfo {year} {2016})},\ \Eprint {https://arxiv.org/abs/1312.3039}
  {arXiv:1312.3039 [math.OC]} \BibitemShut {NoStop}%
\end{thebibliography}%


\clearpage

\onecolumngrid
\appendix


\section{Preliminaries}

\subsection{Choi isomorphism}
\label{appendix:Choi}

The Choi–Jamiołkowski isomorphism~\cite{CHOI1975285, JAMIOLKOWSKI1972275} (here we use the ``Choi version'') is an isomorphism between two different representations of CPTP maps. 
The Choi matrix of a linear map $\mathcal{M} : \L(\H^X) \to \L(\H^Y)$ is defined as
\begin{equation}
    \mathsf{M} =  \mathcal{I} \otimes \mathcal{M}(\dketbra{\id}{\id}^{XX})  = \sum_{i, i'} \ketbra{i}{i'} \otimes \M(\ketbra{i}{i'}) \in \L(\H^X \otimes \H^Y),
\end{equation}
where $\dket{\id}^{XX} = \sum_i \ket{i} \otimes \ket{i} \in \H^{XX}$ is the (unormalised) maximally entangled state, $\{\ket{i}\}_i$ the computational basis of $\H^X$, and $\mathcal{I}:\L(\H^X) \to \L(\H^X)$ the identity channel. 
Here, as throughout, we use the shorthand notation $\H^{XY}:=\H^X\otimes\H^Y$, etc.
Several important properties of a linear map $\M$ can be directly described through its Choi matrix $\mathsf{M}$.
In particular, $\M$ is completely positive (CP) if and only if $\mathsf{M}$ is positive semidefinite (i.e., $\mathsf{M}\geq 0$), and $\M$ is trace preserving (TP) if and only if $\Tr_Y\mathsf{M} = \id^X$. 
These two properties make the Choi isomorphism a useful tool for optimising over quantum channels in semidefinite programs (SDPs).
Finally, the action of $\M$ can be described in terms of its Choi matrix via the inverse Choi isomorphism as
\begin{equation}
\label{eq:actionChoi}
\M(\rho) =\Tr_X\big[(\rho^\mathsf{T} \otimes \id^Y) \mathsf{M}\big],
\end{equation}
where $\cdot ^\mathsf{T}$ is the transpose, and $\Tr_X$ denotes the partial trace over the system $X$.

It is often convenient to manipulate unitary (or isometric) channels at the level of their so-called Choi vectors. 
For an isometric channel $\U$, corresponding to the isometry $U:\H^X\to\H^Y$, its Choi vector is given as
\begin{equation}
    \dket{U} = \id^X \otimes U (\dket{\id}^{XX}) = \sum_i \ket{i} \otimes U\ket{i} \in \H^{XY},
\end{equation}
and its Choi matrix can be recovered as $\mathsf{U} = \dketbra{U}{U}$.


\subsection{Link product}
\label{appendix:link_prod}

The link product \cite{chiribella_quantum_2008, chiribella_theoretical_2009} is a useful tool allowing the composition of CP maps, potentially over a subset of their input/output systems, to be computed directly in the Choi picture.
It is defined, for any matrices $\mathsf{M}^{XY}\in\L(\H^{XY})$ and $\mathsf{N}^{YZ} \in \L(\H^{YZ})$, as
\begin{equation}
\label{eq:link}
    \mathsf{M}^{XY}*\mathsf{N}^{YZ} =\Tr_Y\big[(\mathsf{M}^{XY} \otimes \id^Z)^{\mathsf{T}_Y}(\id^X \otimes \mathsf{N}^{YZ})] \in \L(\H^{XZ}),
\end{equation}
where $\cdot^{\mathsf{T}_Y}$ denotes the partial transpose over the Hilbert space $\H^Y$ with respect to the computational basis. 
Because of the invariance of the trace under cyclic permutations, the link product is commutative, $\mathsf{M}^{XY} * \mathsf{N}^{YZ} = \mathsf{N}^{YZ} * \mathsf{M}^{XY}$, while in $n$-fold link products it is also associative as long as each Hilbert space appears at most twice in the product.
This will notably allow us to write unambiguously, for example, $\mathsf{M}_1 * \cdots * \mathsf{M}_T$, and is likewise why we differentiate, for instance, between the Hilbert spaces $\H^{\alpha_t}$ and $\H^{\alpha_{t'}}$ of the ancillary systems at different places in a quantum circuit (or QC-supermap), even if these spaces are isomorphic.

When $\H^Y$ is trivial (so that $\mathsf{M}^{XY}$ and $\mathsf{N}^{YZ}$ do not share any Hilbert spaces), the link product reduces to the tensor product as $\mathsf{M}^X * \mathsf{N}^Z = \mathsf{M}^X \otimes \mathsf{N}^Z$. 
If, on the other hand, $\H^X$ and $\H^Z$ are trivial (so that the operations act on the same Hilbert space $\H^Y$), then $\mathsf{M}^Y * \mathsf{N}^Y = \Tr[(\mathsf{M}^Y)^\mathsf{T} \mathsf{N}^Y]$.

A version of the link product can also be formulated for vectors, allowing the composition of unitaries (or, more generally, isometries) to be described directly on their Choi vectors~\cite{wechs_quantum_2021}.
The link  product between two vectors $\dket{U} \in \H^{XY}$ and $\dket{V} \in \H^{YZ}$ is given as
\begin{equation}
    \dket{U} * \dket{V} = (\dket{U}^{\mathsf{T}_Y} \otimes \id^Z) \dket{V} \in \H^{XZ}.
\end{equation}


\section{Characterisation of general supermaps}
\label{annex:Gen}

Recall that, as described in the main text, a general $T$-input quantum supermap $\mathcal{S}$ is a $T$-linear completely CP-preserving and TP-preserving map.
Such a supermap can be represented in the Choi picture by its process matrix $W \in\L(\H^{PA^{IO}_{\T}F})$~\cite{araujo_witnessing_2015,Oreshkov2012}. 
Let us write the set of all process matrices of general $T$-input quantum supermaps\footnote{For a given (in general, implicitly) scenario, i.e., specification of $T$ and the dimensions of the Hilbert spaces $\H^P,\H^{A_k^{I}},\H^{A_k^O},\H^F$.} as $\W^{\text{Gen}}$.

The process matrices of quantum supermaps are positive semidefinite matrices normalised to satisfy $\Tr W = d^O$ with $d^O=d^P \prod_{k=1}^T d^O_k$ (where $d^P = \dim(\H^P)$ and $d^O_k = \dim(\H^{A^O_k})$), and which belong to a specific linear subspace $\L^{\text{Gen}}$.
This subspace can be nicely characterised through linear constraints on $W^\text{Gen}$, allowing one to optimise over general supermaps in SDPs.
This characterisation can be nicely expressed using the ``trace-out-and-replace'' notation~\cite{araujo_witnessing_2015} defined as
\begin{equation}
    {}_X W = \frac{\id^X}{d_X} \otimes \Tr_X W \quad \text{and} \quad
    {}_{[1-X]} W = W - \frac{\id^X}{d_X} \otimes \Tr_X W,
\end{equation}
with $d_X = \text{dim}(\H^X)$.
Let us also define the ``reduced'' process matrices $W^{[P]}=\Tr_{A^{IO}_{\T}F}W\in\L(\H^P)$ and, for any nonempty subset $\K \subseteq \T$,
\begin{equation}
W^{[P \mathcal{K}]} = \Tr_{A^{IO}_{\mathcal{T} \setminus \mathcal{K}}F} W \in \L(\H^{PA^{IO}_\mathcal{K}}).
\end{equation}

\begin{proposition}[From \cite{araujo_witnessing_2015,wechs_quantum_2021}]
	\label{prop:GenCharachterisation}
	A matrix $W\in\L(\H^{PA^{IO}_{\T}F})$ is the process matrix of a $T$-input quantum supermap (i.e., $W\in \W^{\textup{Gen}}$) if and only if $W$ is positive semidefinite ($W\ge 0$), $\Tr W = d^O$, and $W\in \L^{\textup{Gen}}$, where the linear subspace $\L^{\textup{Gen}}$ is defined as
	\begin{align}
		W \in \L^{\textup{Gen}} \iff \forall \emptyset \subsetneq\mathcal{K} \subseteq \T,\ {}_{\Pi_{k \in \mathcal{K}} [1 - A^O_k]} W^{[P\mathcal{K}]} = 0 \text{ and } {}_{[1 - P]}W^{[P]} = 0.
	\end{align}
\end{proposition}

\section{Characterisation of FO-supermaps}
\label{appendix:FO}

In a $T$-input fixed order quantum supermap, the input operations are applied in a fixed sequential order. 
For example, the order $(1, \dots, T)$ means that the input operations are applied in the order they are passed as arguments. 
Such a supermap can be represented in the Choi picture by its process matrix $W \in\L(\H^{PA^{IO}_{\T}F})$~\cite{chiribella_theoretical_2009,chiribella_quantum_2008}, which can be decomposed as the link product of some quantum channels, $W = \mathsf{V}_{T+1} * \dots * \mathsf{V}_1$, with the $\mathsf{V}_i$'s, for $2 \leq t \leq T$, being the Choi matrices of quantum channels $\V_t : \L(\H^{A^O_t \alpha_t}) \to \L(\H^{A^I_{t+1} \alpha_{t+1}})$, and $\V_1 : \L(\H^P) \to \L(\H^{A^I_{1} \alpha_1})$ and $\V_{T+1} : \L(\H^{A^O_{T} \alpha_T}) \to \L(\H^F)$, where the $\H^{\alpha_t}$ are ancillary Hilbert spaces and can, without loss of generality, be taken to be isomorphic (we differentiate them for clarity and to write unambiguously the link product of all $\mathsf{V}_t$ as above). 
Let us write the set of all process matrices of $T$-input FO-supermaps as $\W^{\text{FO}}$.

A nice characterisation of the set $\W^{\text{FO}}$ is obtained by combining the trace-preserving conditions on the Choi matrices of the $\mathsf{V}_t$ and the decomposition of $W$ into their link product.

\begin{proposition}[From \cite{wechs_quantum_2021, chiribella_quantum_2008}]
A matrix $W \in\L(\H^{PA^{IO}_{\T}F})$ is the process matrix of a $T$-input quantum supermap with fixed order $(1, \dots, T)$ (i.e, $W \in \W^{\text{FO}}$) if and only if it is positive semidefinite and its reduced matrices, defined for $1 \leq t \leq T$ as
\begin{equation}
\label{eq:FO1}
    W_{(t)} = \frac{1}{d^O_t d^O_{t+1} \cdots d^O_T}\Tr_{A^O_tA^{IO}_{\{t+1, \dots, T\}}F}W \in \L(\H^{PA^{IO}_{\{1, \dots, t-1\}}A^I_t}),
\end{equation}
satisfy
\begin{equation}
\label{eq:FO2}
\begin{aligned}
    & \Tr_{A^I_1} W_{(1)} = \id^P, \\
    &\forall n =1, \dots, T - 1, \ \Tr_{A^I_{t+1}} W_{(t+1)} = W_{(t)} \otimes \id^{A_t^O}, \\
    &\text{and} \ \Tr_F W = W_{(T)} \otimes \id^{A_T^O}.
\end{aligned}
\end{equation}
\end{proposition}

Note that the linear constraints~\eqref{eq:FO1} and \eqref{eq:FO2} both enforce the normalisation contraints $\Tr W = d^O$ and define a linear subspace $\mathcal{L}^\text{FO}$, so that, in analogy to the characterisation of $\mathcal{W}^\text{Gen}$ in Proposition~\ref{prop:GenCharachterisation}, a matrix $W\in\L(\H^{PA^{IO}_{\T}F})$ is in $\mathcal{W}^\text{FO}$ if and only if $W\ge 0$, $\Tr W = d^O$ and $W\in \mathcal{L}^\text{FO}$~\cite{Wechs_2019}.


\section{Description and characterisation of QC-supermaps}
\label{appendix:QC}

Supermaps with quantum control of causal order (or QC-supermaps) are a form of generalised quantum circuits where a quantum system is used to coherently control the order in which the input operations are applied.
For simplicity, we will assume throughout the presentation of QC-supermaps here that the input operations are unitaries $(U_1,\dots,U_T)$.
For a more detailed presentation and the (rather simple) generalisation to generic channels as input operations (which leads to the same set of corresponding process matrices), see~\cite{wechs_quantum_2021}.

A QC-supermap alternates between applying some controlled ``internal'' operations on a ``target'' system and some ancillary systems, potentially correlating these systems with the quantum control, and using this control system to control which input operations to apply at time-step $t$.
More precisely, throughout the computation the control system takes, at time-step $t$, the basis states  $\ket{\K_{t-1},k_t}^{C_{t}^{(\prime)}}\in\H^{C_t^{(\prime)}}$, specifying that operation $k_t\notin \K_{t-1}$ will be applied and that the operations $\K_{t-1}\subsetneq\T$ have already been applied.
This control system controls the application of the different $U_{k_t}$ and the pure operations $V_{\K_{t-1},k_t}^{\to k_{t+1}}$ between time-steps, and is the minimal control system required to ensure that no operation is applied more than once, while maintaining the possibility for superpositions of causal orders and interference of different causal histories (i.e., permutations of the operations in $\K_{t-1}$).

A QC-supermap is completely specified by the composants of the internal operations, $V_{\emptyset,\emptyset}^{\to k_1}:\H^P \to \H^{A_{k_1}^I\alpha_1}$, $V_{\K_{t-1},k_t}^{\to k_{t+1}}: \H^{A_{k_t}^O\alpha_t}\to\H^{A_{k_{t+1}}^I\alpha_{t+1}}$ (for $1\le t \le T-1$) and $V_{\K_{T-1},k_T}^{\to F}:\H^{A_{k_T}^O\alpha_T}\to\H^{F\alpha_F}$, where the $\H^{\alpha_t}$ and $\H^{\alpha_F}$ are again ancillary spaces which, without loss of generality, can be assumed to be isomorphic at each time-step.
In order to formally express the control of both the internal and external (input) operations -- which act on \emph{a priori} different spaces -- it is necessary to embed these operations in some generic input and output spaces $\H^{\tilde{A}_t^I}$ and $\H^{\tilde{A}_t^O}$ for each time-step $t$~\cite{wechs_quantum_2021}.
We denote the corresponding operators $\tilde{U}_{k_t}^{[t]}:\H^{\tilde{A}_{t}^I}\to\H^{\tilde{A}_{t}^O}$ (for each $k_t\in\T$), $\tilde{V}_{\emptyset,\emptyset}^{\to k_1}:\H^P \to \H^{\tilde{A}_{1}^I\alpha_1}$, $\tilde{V}_{\K_{t-1},k_t}^{\to k_{t+1}}: \H^{\tilde{A}_{t}^O\alpha_t}\to\H^{\tilde{A}_{t+1}^I\alpha_{t+1}}$ (for $1\le t \le T-1$) and $\tilde{V}_{\K_{T-1},k_T}^{\to F}:\H^{\tilde{A}_{T}^O\alpha_T}\to\H^{F\alpha_F}$.

A QC-supermap thus proceeds by applying, at each time-step $t$, the coherently controlled input operations 
\begin{equation}
\label{eq:embedding}
	\tilde{U}^{[t]}=\sum_{\K_{t-1},k_t}\tilde{U}_{k_t}^{[t]}\otimes\ket{\K_{t-1},k_t}^{C_t'}\bra{\K_{t-1},k_t}^{C_t},
\end{equation}
while between time-steps it applies the internal operations,
\begin{align}
\label{eq:genQC}
    \tilde{V}_1 &= \sum_{k_1} \tilde{V}^{\to k_1}_{\emptyset, \emptyset} \otimes \ket{\emptyset, k_1}^{C_1}: \H^P \to \H^{\tilde{A}_1^I\alpha_1 C_1}, \\
    \tilde{V}_{t+1} &= \sum_{\substack{\K_{t-1}, \\ k_t, k_{t+1}}} \tilde{V}^{\to k_{t+1}}_{\K_{t-1}, k_t} \otimes \ket{ \K_{t-1} \cup k_t, k_{t+1}}^{C_{t+1}}\bra{\K_{t-1},k_t}^{C_t'} : \H^{\tilde{A}^O_t\alpha_t C'_t} \to \H^{\tilde{A}^I_{t+1}\alpha_{t+1} C_{t+1}}, \\
\label{eq:genQC2}
     \tilde{V}_{T+1} &= \sum_{k_T} \tilde{V}^{\to F}_{\T \setminus k_T, k_T} \otimes \bra{\T \setminus k_T, k_T}^{C_T'}: \H^{\tilde{A}_T^O\alpha_T C'_T} \to \H^{F \alpha_F},
\end{align}
which are required to be pure isometries.

The process matrix is obtained from the internal operations as 
\begin{align}\label{eq:WfromProcessVector}
    W &= \Tr_{\alpha_F}\! \ketbra*{w_{(\T, F)}}{w_{(\T, F)}} \in \L(\H^{PA^{IO}_\T F}),
\end{align}
with $\ket*{w_{(\T, F)}} = \sum_{(k_1, \dots, k_T)} \ket*{w_{(k_1, \dots, k_T, F)}}$ and
\begin{equation}\label{eq:processVector}
\ket*{w_{(k_1, \dots, k_T, F)}} = \dket{V^{\to k_1}_{\emptyset, \emptyset}} * 
\dket{V^{\to k_2}_{\emptyset, k_1}} * 
\dket{V^{\to k_3}_{\{k_1\}, k_2}} * \dots * 
\dket{V^{\to k_T}_{\{k_1, \dots, k_{T-2}\}, k_{T-1}}} * 
\dket{V^{\to F}_{\{k_1, \dots, k_{T-1}\}, k_T} } 
\in \H^{PA^{IO}_\T F \alpha_F}.
\end{equation}
Let us write the set of all process matrices of $T$-input QC-supermaps as $\W^{\text{QC}}$.

One can readily verify that the action of a QC-supermap $\S$ on input operations $(U_1,\dots,U_T)$, as calculated through the process matrix, indeed coincides with that obtained by applying iteratively the internal operations $\tilde{V}_t$ and controlled operations $\tilde{U}^{[t]}$ in the generic input and output spaces.
That is, 
\begin{align}
	\mathsf{S}(\mathsf{U}_1,\dots,\mathsf{U}_T) &= (\mathsf{U}_1\otimes \cdots \mathsf{U}_T)*W \\
    \label{eq:QCAction}
	&= \Tr_{\alpha_F}[\tilde{\mathsf{V}}_{T+1} * \tilde{\mathsf{U}}^{[T]} * \tilde{\mathsf{V}}_T * \cdots * \tilde{\mathsf{V}}_2 * \tilde{\mathsf{U}}^{[1]} * \tilde{\mathsf{V}}_1].
\end{align}

As for the other classes of supermaps we presented, one can obtain a nice characterisation of $\W^{\text{QC}}$ from the requirement for the internal operations~\eqref{eq:genQC} to be trace-preserving isometries.

\begin{proposition}[From \cite{wechs_quantum_2021}]
A matrix $W \in \L(\H^{PA^{IO}_{\T}F})$ is the process matrix of a $T$-input QC-supermap (i.e, $W \in \mathcal{W}^{\textup{QC}}$) if and only if there exist some positive semidefinite matrices $W_{(\K_{t-1}, k_t)} \in \L(\H^{PA^{IO}_{\K_{t-1}}A^{I}_{k_t}})$, for all strict subsets $\K_{t-1}$ of $\T$ and all $k_t \in \T \setminus \K_{t-1}$, satisfying
\begin{equation}
\label{eq:QC}
    \begin{aligned}
    &\sum_{k_1 \in \T} \Tr_{A^I_{k_1}} W_{(\emptyset, k_1)} = \id^P, \\
    &\forall \emptyset \subsetneq \K_t \subsetneq \T, \ \sum_{k_{t+1} \in \T \setminus \K_t} \Tr_{A^I_{k_{t+1}}} W_{(\K_t, k_{t+1})} = \sum_{k_t \in \K_t} W_{(\K_t \setminus k_t, k_t)} \otimes \id^{A^O_{k_t}}, \\
    &\text{and} \  \Tr_F W = \sum_{k_T \in \T} W_{(\T \setminus k_T, k_T)} \otimes \id^O_{k_T}.
    \end{aligned}
\end{equation}
\end{proposition}

\section{Proof of Theorem~\ref{th:QC-S}}
\label{appendix:equiv}

In this appendix, we give a proof of Theorem~\ref{th:QC-S}. 
This theorem states that for any $T$-input QC-supermap, one can define a $T$-input FO-supermap such that their action on $T$ copies of any unitary $U$ is the same.
\QC*
\noindent
\begin{proof}
Let us consider a $T$-input QC-supermap $\mathcal{S}^{\text{QC}}$, as well as $T$ copies of some unitary $U$ labelled from 1 to $T$ as $(U^{(1)}, \dots, U^{(T)})$, with $U^{(t)}: \H^{A_t^I} \to \H^{A_t^O}$ (for $1 \leq t \leq T$) to distinguish which Hilbert spaces they act on.
We emphasise that all the $\H^{A_t^I}$ and $\H^{A_t^O}$ are isomorphic, and the $U^{(t)}$ have formally the same action on their respective spaces. 
Following Eqs.~\eqref{eq:genQC}--\eqref{eq:genQC2}, the QC-supermap $\mathcal{S}^{\text{QC}}$ can be represented as a generalised quantum circuit whose internal operations are defined on some generic input and output spaces $\H^{\tilde{A}_t^I}$ and $\H^{\tilde{A}_t^O}$ (which again are here all isomorphic to the $\H^{A_{k_t}^I}$ and $\H^{A_{k_t}^O}$) for each time-step $t$:

\begin{align}
\label{eq:inner}
    \tilde{V}_1 &= \sum_{k_1} \tilde{V}^{\to k_1}_{\emptyset, \emptyset} \otimes \ket{\emptyset, k_1}^{C_1}: \H^P \to \H^{\tilde{A}_1^I\alpha_1 C_1}, \\
    \tilde{V}_{t+1} &= \sum_{\substack{\K_{t-1}, \\ k_t, k_{t+1}}} \tilde{V}^{\to k_{t+1}}_{\K_{t-1}, k_t} \otimes \ket{ \K_{t-1} \cup k_t, k_{t+1}}^{C_{t+1}}\bra{\K_{t-1},k_t}^{C_t'} : \H^{\tilde{A}^O_t\alpha_t C'_t} \to \H^{\tilde{A}^I_{t+1}\alpha_{t+1} C_{t+1}}, \\
    \label{eq:inner2}
     \tilde{V}_{T+1} &= \sum_{k_T} \tilde{V}^{\to F}_{\T \setminus k_T, k_T} \otimes \bra{\T \setminus k_T, k_T}^{C_T'}: \H^{\tilde{A}_T^O\alpha_T C'_T} \to \H^{F \alpha_F}.
\end{align}

From these internal operations let us define an FO-supermap $\S^{\text{FO}}$, independent of $U$, whose action on $T$ copies of any unitary $U$ will be the same as that of the QC-supermap $\S^{\text{Gen}}$. 
Recall (see Appendix~\ref{appendix:FO}) that, to define $\S^{\text{FO}}$, it suffices to specify the $T+1$ quantum channels (or, equivalently, their Choi matrices) $\V_t$ (for $1\le t \le T$) and $\V_{T+1}$ which are applied in alternation with the input operations.
To this end, let us consider the 
operations $V'^{\to k_1}_{\emptyset, \emptyset} : \H^{P}\to \H^{A_1^I\alpha_1}$, $V'^{\to k_{t+1}}_{\K_{t-1}, k_t} : \H^{A_t^O\alpha_t}\to\H^{A_{t+1}^I\alpha_{t+1}}$ and $V'^{\to F}_{\T \setminus k_T, k_T}:\H^{A_T^O\alpha_T}\to\H^{F}$ which have the same action as the corresponding operations in Eqs.~\eqref{eq:inner}--\eqref{eq:inner2} but act now on the non-tilded spaces $\H^{A_t^I} \cong \H^{\tilde{A}_t^I}$ and $\H^{A_t^O} \cong \H^{\tilde{A}_t^O}$ (for $1 \leq t \leq T)$.
Note that this is possible precisely because all the input and output spaces are isomorphic.
By relabelling the control spaces $C_t'$ to $C_t$ (since, in a FO-supermap, the input operations are not controlled and there is thus no need to distinguish the control system before and after it is used), we hence define the isometries
\begin{align}
\label{eq:innerFO}
        V'_1 &= \sum_{k_1} V'^{\to k_1}_{\emptyset, \emptyset} \otimes \ket{\emptyset, k_1}^{C_1}: \H^P \to \H^{A_1^I\alpha_1 C_1}, \\
        V'_{t+1} &= \sum_{\substack{\K_{t-1}, \\ k_t, k_{t+1}}}  V'^{\to k_{t+1}}_{\K_{t-1}, k_t} \otimes \ket{ \K_{t-1} \cup k_t, k_{t+1}}^{C_{t+1}}\bra{\K_{t-1},k_t}^{C_t} : \H^{A^O_t\alpha_t C_t} \to \H^{A^I_{t+1}\alpha_{t+1} C_{t+1}}, \\
        \label{eq:innerFO2}
        \hat{V}'_{T + 1} &= \sum_{k_T}  V'^{\to F}_{\T \setminus k_T, k_T} \otimes \bra{\T \setminus k_T, k_T}^{C_T} : \H^{A_T^O\alpha_T C_T} \to \H^{F\alpha_F}.
\end{align}
For the FO-supermap we then reinterpret the control systems $\H^{C_t}$ as part of the ancillary spaces, taking $\H^{\alpha_t'}:=\H^{\alpha_t C_t}$, so that the Choi matrices of the channels defining $\S^{\text{FO}}$ are $\mathsf{V}_t'=\dketbra{V_t'}{V_t'}$ (for $1\le t \le T$) and $\mathsf{V}_{T+1}'=\Tr_{\alpha_F}\dketbra{\hat{V}_{T+1}'}{\hat{V}_{T+1}'}$.
The fact that these are indeed valid channels follows from the fact that Eqs.~\eqref{eq:inner}--\eqref{eq:inner2} specify valid isometries and the isomorphisms between the generic (tilded) and specific (non-tilded) input and output spaces.
The action of $\S^{\text{FO}}$ on $T$ copies of $U$, according to Eq.~\eqref{eq:ActionFO}, is
\begin{equation}
\label{eq:actionFO_U}
\mathsf{S}^{\text{FO}}(\mathsf{U}^{(1)}, \dots, \mathsf{U}^{(T)}) = \mathsf{V}'_{T+1}* \mathsf{U}^{(T)} * \dots *\mathsf{V}'_2 * \mathsf{U}^{(1)} *\mathsf{V}'_1.
\end{equation}

We now show that the action of the QC-supermap $\S^{\text{QC}}$ on $T$-copies of $U$ is equivalent to that of $\S^{\text{FO}}$ given by Eq.~\eqref{eq:actionFO_U}. 
Recall (see Appendix~\ref{appendix:QC}) that, at each time-step $t$, $\S^{\text{QC}}$ proceeds by applying the coherently controlled operations $\tilde{U}^{[t]}=\sum_{\K_{t-1},k_t}\tilde{U}_{k_t}^{[t]}\otimes\ket{\K_{t-1},k_t}^{C_t'}\bra{\K_{t-1},k_t}^{C_t}$, where $\tilde{U}^{[t]}_{k_t}: \H^{\tilde{A}_{t}^I}\to\H^{\tilde{A}_{t}^O}$ (for $1 \leq k_t \leq T$) is an embedding of $U^{(k_t)}$ into the generic spaces.
Since each $U^{(k_t)}$ has formally the same action, all the $\tilde{U}_{k_t}^{[t]}$ (for a given $t$) are in fact identically the same operation; let us denote these unitary $\tilde{U}^{(t)}:\H^{\tilde{A}_t^I}\to\H^{\tilde{A}_t^O}$.
The coherently controlled operations $\tilde{U}^{[t]}:\H^{\tilde{A}_{t}^I C_t}\to\H^{\tilde{A}_{t}^O C_t'}$ can thus be written in the factorised form
\begin{align}
    \label{eq:abuse}
	\tilde{U}^{[t]} &= \tilde{U}^{(t)} \otimes \sum_{\K_{t-1},k_t} \ket{\K_{t-1},k_t}^{C_t'}\bra{\K_{t-1},k_t}^{C_t} \\
    \label{eq:fact}
    &=  \tilde{U}^{(t)} \otimes \id^{C_t \to C'_t}.
\end{align}
Then following Eq.~\eqref{eq:QCAction}, the action of the QC-supermap on the $T$ copies of $U$ is 
\begin{align}
	\mathsf{S}^{\text{QC}}(\mathsf{U}^{(1)},\dots,\mathsf{U}^{(T)}) &= \Tr_{\alpha_F}[\mathsf{\tilde{V}}_{T+1} * \mathsf{\tilde{U}}^{[T]} * \mathsf{\tilde{V}}_T * \cdots * \mathsf{\tilde{V}}_2 * \mathsf{\tilde{U}}^{[1]} * \mathsf{\tilde{V}}_1] \\
    \label{decomp}
    &= \Tr_{\alpha_F}[\mathsf{\tilde{V}}_{T+1} * (\mathsf{\tilde{U}}^{(T)} * \dketbra{\id}{\id}^{C_t \to C_t'}) * \mathsf{\tilde{V}}_T * \cdots * \mathsf{\tilde{V}}_2 * (\mathsf{\tilde{U}}^{(1)} * \dketbra{\id}{\id}^{C_1 \to C_1'}) * \mathsf{\tilde{V}}_1] \\
    \label{comm}
    &= \Tr_{\alpha_F}[\mathsf{\tilde{V}}_{T+1} * \dketbra{\id}{\id}^{C_t \to C_t'} * \mathsf{\tilde{U}}^{(T)} * \mathsf{\tilde{V}}_T * \cdots * \mathsf{\tilde{V}}_2 * \dketbra{\id}{\id}^{C_1 \to C_1'} * \mathsf{\tilde{U}}^{(1)}  * \mathsf{\tilde{V}}_1] \\
    &= \Tr_{\alpha_F}[\mathsf{\tilde{V}}_{T+1} * \dketbra{\id}{\id}^{C_t \to C_t'}] * \mathsf{\tilde{U}}^{(T)} * \mathsf{\tilde{V}}_T * \cdots * \mathsf{\tilde{V}}_2 * \dketbra{\id}{\id}^{C_1 \to C_1'} * \mathsf{\tilde{U}}^{(1)}  * \mathsf{\tilde{V}}_1,
\end{align}
where the second line follows from Eq.~\eqref{eq:fact}, the third from the commutation of the link product, and the fourth from the fact that only $\mathsf{\tilde{V}}_{T+1}$ is defined on $\H^{\alpha_F}$. 
Now by defining $\mathsf{\tilde{V}}_{T+1}' = \Tr_{\alpha_F}[\mathsf{\tilde{V}}_{T+1} * \dketbra{\id}{\id}^{C_t \to C_t'}]$, $\mathsf{\tilde{V}}_{t+1}' = \mathsf{\tilde{V}}_{t+1} * \dketbra{\id}{\id}^{C_t \to C_t'}$ (for $1 \leq t \leq T-1$) and $\mathsf{\tilde{V}}_1' = \mathsf{\tilde{V}}_1$, we have
\begin{equation}
	\mathsf{S}^{\text{QC}}(\mathsf{U}^{(1)},\dots,\mathsf{U}^{(T)}) = \mathsf{\tilde{V}}'_{T+1}* \mathsf{\tilde{U}}^{(T)} * \dots *\mathsf{\tilde{V}}'_2 * \mathsf{\tilde{U}}^{(1)} *\mathsf{\tilde{V}}'_1.
\end{equation}
For $1 < t \leq T$, the link product $\mathsf{\tilde{V}}_{t+1} * \dketbra{\id}{\id}^{C_t \to C_t'}$ is the composition of $\tilde{V}_{t+1}: \H^{\tilde{A}^O_t\alpha_t C'_t} \to \H^{\tilde{A}^I_{t+1}\alpha_{t+1} C_{t+1}}$ and an identity channel between $\H^{C_t}$ and $\H^{C_t'}$, which effectively relabels $\H^{C_t'}$ to $\H^{C_t}$ in $\mathsf{\tilde{V}}_{t+1}$. 
Finally, because the spaces $\H^{\tilde{A}^I_t}$ and $\H^{\tilde{A}^O_t}$ are isomorphic to $\H^{A^I_t}$ and $\H^{A^O_t}$, we have that 
\begin{align}
	\mathsf{S}^{\text{QC}}(\mathsf{U}^{(1)},\dots,\mathsf{U}^{(T)}) &= \mathsf{V}'_{T+1}* \mathsf{U}^{(T)} * \dots *\mathsf{V}'_2 * \mathsf{U}^{(1)} *\mathsf{V}'_1 \\
 &= \mathsf{S}^{\text{FO}}(\mathsf{U}^{(1)}, \dots, \mathsf{U}^{(T)}),     
\end{align}
as desired.
\end{proof}

One way to understand more intuitively the FO-supermap $\S^\text{FO}$ defined above from a $\S^\text{QC}$ is to return to the expression~\eqref{eq:WfromProcessVector} writing the process matrix $W^\text{QC}$ of $\S^\text{QC}$ as $W^\text{QC}=\Tr_{\alpha_F}\!\ketbra*{w_{(\T,F)}}{w_{(\T,F)}}$ where $\ket{w_{(\T,F)}}=\sum_{(k_1,\dots,k_T)}\ket*{w_{(k_1\dots,k_T,F)}}$ with $\ket*{w_{(k_1\dots,k_T,F)}}$ defined as in Eq.~\eqref{eq:processVector}.
The process matrix $W^\text{FO}$ of $\S^\text{FO}$ can be then seen to be given as $W^\text{FO}=\Tr_{\alpha_F}\!\ketbra*{w'_{(\T,F)}}{w'_{(\T,F)}}$
where $\ket*{w'_{(\T,F)}}=\sum_{(k_1,\dots,k_T)}\ket*{w'_{(k_1\dots,k_T,F)}}$ and where the vectors $\ket*{w'_{(k_1\dots,k_T,F)}}$ are the same as $\ket*{w_{(k_1\dots,k_T,F)}}$ but with each space $\H^{A^{IO}_{k_t}}$ now relabelled as $\H^{A^{IO}_t}$ so that each $\ket*{w'_{(k_1\dots,k_T,F)}}$ corresponds to the fixed order $(1,\dots,T)$ instead of $(k_1,\dots,k_T)$.

Note that while the implementation of a QC-supermap is not unique (i.e., different choices of internal operations $V_{\K_{t-1},k_t}^{\to k_{t+1}}$ may give the same process matrix $W^\text{QC}$), the process matrix $W^\text{FO}$ obtained from the above mapping in general depends on the implementation one takes.
Consider for example the QC-supermap with process matrix $W=\ketbra{w}{w}$ where 
\begin{equation}
    \ket{w} = \ket{\psi_1}^{A^I_1}\ket{\psi_2}^{A^I_2} \ket{\id}^{A^O_1F_1}\ket{\id}^{A^O_2 F_2},
\end{equation}
with $F:=F_1F_2$ and implicit tensor products.
This supermap is compatible with both possible fixed orders, $(1,2)$ and $(2,1)$ (i.e., it can be implemented in parallel).
Two non-trivial implementations as QC-supermaps can be obtained by taking $\dket{\bar{V}_{\emptyset,\emptyset}^{\to 1}}=\ket{\psi_1}^{A^I_1}$, $\dket{\bar{V}_{\emptyset,1}^{\to 2}}=\dket{\id}^{A^O_1\alpha}\ket{\psi_2}^{A_2^I}$ and $\dket{\bar{V}_{\{1\},2}^{\to F}}=\dket{\id}^{\alpha F_1}\dket{\id}^{A_2^O F_2}$ or $\dket{\bar{\bar{V}}_{\emptyset,\emptyset}^{\to 2}}=\ket{\psi_2}^{A^I_2}$, $\dket{\bar{\bar{V}}_{\emptyset,2}^{\to 1}}=\dket{\id}^{A^O_2\alpha}\ket{\psi_1}^{A_1^I}$ and $\dket{\bar{\bar{V}}_{\{2\},1}^{\to F}}=\dket{\id}^{\alpha F_2}\dket{\id}^{A_1^O F_1}$, which gives
\begin{align}
	\ket{\bar{w}} &= \ket{w_{(1,2,F)}} = \dket{\bar{V}_{\emptyset,\emptyset}^{\to 1}} * \dket{\bar{V}_{\emptyset,1}^{\to 2}} * \dket{\bar{V}_{\{1\},2}^{\to F}}\\
	\ket{\bar{\bar{w}}} &= \ket{w_{(2,1,F)}} = \dket{\bar{\bar{V}}_{\emptyset,\emptyset}^{\to 2}} * \dket{\bar{\bar{V}}_{\emptyset,2}^{\to 1}} * \dket{\bar{\bar{V}}_{\{2\},1}^{\to F}}.
\end{align}
One can immediately verify that $\ket{w}=\ket{\bar{w}}=\ket{\bar{\bar{w}}}$, so that these implementations define the same supermap.
However, applying the mapping used in the proof of Theorem~\ref{th:QC-S} to these two implementations, one obtains two FO-supermaps defined by the process vectors
\begin{align}
	\ket{\bar{w}'} &= \ket{\psi_1}^{A^I_1}\ket{\psi_2}^{A^I_2} \ket{\id}^{A^O_1F_1}\ket{\id}^{A^O_2 F_2}\\
	\ket{\bar{\bar{w}}'} &= \ket{\psi_2}^{A^I_1}\ket{\psi_1}^{A^I_2} \ket{\id}^{A^O_1F_2}\ket{\id}^{A^O_2F_1},
\end{align}
which define two distinct supermaps that, nonetheless, indeed will always have the same action (as each other and the original QC-supermap $W^\text{QC}$) when applied to two copies of the same unitary channel.


\section{Proof of Theorem~\ref{th:bound_process}}
\label{appendix:TheoremPol}

In this appendix, we prove a generalisation of the polynomial bound from FO-supermaps (or quantum circuits) to general supermaps.
\boundProcess*
\noindent
The proof is similar to that of the original statement of the polynomial bound for FO-supermaps~\cite{beals_quantum_2001}. Let us first prove the following lemma.
\begin{lemma}
	\label{lemma}
For $x \in \{0, 1\}^n$, $\dket{O_x^{\otimes T}}$ is a vector whose coefficients are multivariate polynomials in $x$ of degree at most $T$.
\end{lemma}
\begin{proof}
	We proceed by induction.
If $T = 1$ we have
\begin{equation}
    \dket{O_x} = \sum_i \ket{i} \otimes O_x \ket{i} = \sum_i (-1)^{x_i} \ket{i} \otimes \ket{i} = \sum_i (1 - 2x_i) \ket{i} \otimes \ket{i},
\end{equation}
which is indeed a multivariate polynomial of degree 1.
Let us note that, for any $T\ge 1$, $\dket{O_x^{\otimes T}}$ is of the form $\dket{O_x^{\otimes T}}=\sum_z \alpha_z(x)\ket{z}\otimes\ket{z}$.
Now let us suppose, for some $T$, that the coefficients $\alpha_z(x)$ are multivariate polynomials of degree at most $T$.
Then
\begin{align}
  \dket{O_x^{\otimes (T+1)}} &=\dket{O_x ^{\otimes T}} \otimes \dket{O_x}\\ &= \sum_z \alpha_z(x) \ket{z} \otimes \ket{z} \otimes \sum_i (1 - 2x_i) \ket{i} \otimes \ket{i} \\
  &= \sum_{z,i} (1 - 2x_i)\alpha_z(x) \ket{z}\otimes\ket{z}\otimes\ket{i}\otimes\ket{i}.
\end{align}
But $(1-2x_i)\alpha_z(x)$ is then a multivariate polynomial of degree at most $T+1$,
which concludes the proof of the lemma.
\end{proof}

The proof of the theorem is now straightforward. Let $f$ be a Boolean function on $n$ bits and $\S$ a supermap with trivial input space and a qubit output space characterised by a process matrix $W$. Let us define, for $x \in \{0, 1\}^n$,
\begin{equation}
    g(x)
    = \Tr[(\mathsf{O}_x^{\otimes T} * W) \cdot \Pi_{1}],
\end{equation}
where $\Pi_1 = \ketbra{1}$.
The function $g$ corresponds to the probability of obtaining the outcome $1$ when measuring the qubit $\mathsf{S}(\mathsf{O}_x^{(1)}, \dots, \mathsf{O}_x^{(T)})$ in the computational basis (cf.\ Eq.~\eqref{eq:computation}). 
Now, it follows from Lemma~\ref{lemma} that $\mathsf{O}_x^{\otimes T} = \dketbra{O_x^{\otimes T}}{O_x^{\otimes T}}$ is a matrix whose coefficient are multivariate polynomials in $x$ of degree at most $2T$, and, since the trace and link product are linear, $g(x)$ is also a multivariate polynomial in $x$ of degree at most $2T$.
If $\mathcal{S}$ computes $f$ then $g(x)=f(x)$ for all $x$, so that $g$ represents $f$, and hence $2Q^{\text{Gen}}_E(f) \ge \deg(f)$.
Similarly, in the bounded error case, if $\mathcal{S}$ computes $f$ with bounded error $\varepsilon=1/3$, then $|g(x)-f(x)|\le 1/3$ for all $x$ and $g$ approximates $f$ with bounded error $\varepsilon=1/3$ also, from which we likewise obtain $2Q^{\text{Gen}}_2(f) \ge \widetilde{\deg}(f)$.


\section{Dual form of the SDP}
\label{Appendix:dual}

In this appendix, we derive the dual form of the primal SDP~\eqref{sdp} computing the minimum error $\varepsilon_T^{\mathcal{C}}(f)$ with which $f$ can be computed by a $T$ query supermap in the class $\mathcal{C}$. 

The dual SDP will be used in Appendix~\ref{appendix:extraction} to convert the numerical solutions of the SDP~\eqref{sdp} into analytic proofs that $\varepsilon_T^{\mathcal{C}}(f)$ lies within a given interval. 
We obtain the dual form using the Lagrangian method and some techniques previously used in Ref.~\cite{PhysRevLett.127.200504} for a slightly different SDP. 

First, for any class $\mathcal{C} \in \{\text{FO, QC, Gen}\}$, by identifying $\L(H^{PA^{IO}_\mathcal{T}F})$ with its dual space, one can define the dual affine space of $W^{\mathcal{C}} \subseteq \L(H^{PA^{IO}_\mathcal{T}F})$ as
\begin{equation}
\label{eq:characDual}
	\overline{\mathcal{W}}^\mathcal{C} = \big\{ \overline{W} \in \L(H^{PA^{IO}_\mathcal{T}F}) \, : \, \Tr[\overline{W} \cdot {W}] = 1 \ \ \forall\, W \in \mathcal{W}^{\mathcal{C}} \big\}.
\end{equation}
Note that the elements $\overline{W}$ of $\overline{\W}^\C$ are normalised such that $\Tr\overline{W} = d^I$, with $d^I = d^F \prod_{k=1}^T d^I_k$ (where $d^F = \dim(\H^F)$ and $d^I_k = \dim(\H^{A^I_k})$). For finite dimensional spaces (as is the case here), because $\W^{\mathcal{C}}$ is an affine set we have $\overline{\overline{\W}^{\mathcal{C}}} = \W^{\mathcal{C}}$.
Taking $\overline{\W}^{\mathcal{C}}_B=\{\overline{W}_k\}_k$ to be an affine basis of $\overline{\W}^\C$, we can thus write $\W^\mathcal{C}$ as
\begin{equation}\label{eq:characDualWithBasis}
	\mathcal{W}^\mathcal{C} = \big\{ W \in \L(H^{PA^{IO}_\mathcal{T}F}) \, : \, \Tr[W \cdot \overline{W}_k] = 1 \ \ \forall\, \overline{W}_k \in \overline{\W}^{\mathcal{C}}_B \big\}.
\end{equation}
Using this characterisation of $\W^\mathcal{C}$ and writing the primal \eqref{sdp} in a different but equivalent form (see~\cite{boyd2004convex}, p.\ 264) one obtains the following canonical form of the primal SDP, 
\begin{equation}
\label{eq:newPrimal}
\begin{split}
 \min_{\varepsilon, W^{[0]}, W^{[1]}}\ &\ - (1 - \varepsilon), \\
\text{s.t.} \ &\  \forall x \in F^{[0]}, \ 1 - \varepsilon - \Tr[W^{[0]} \mathsf{O}_x^{\otimes T}] \leq 0, \\
       &\ \forall x \in F^{[1]}, \ 1 - \varepsilon - \Tr[W^{[1]} \mathsf{O}_x^{\otimes T}] \leq 0, \\
       &\ - W^{[0]} \leq 0, \ - W^{[1]} \leq 0, \ - \varepsilon \leq 0, \\
       &\ \Tr[(W^{[0]} + W^{[1]}) \overline{W}_k] = 1, \ \forall \overline{W}_k \in \overline{\mathcal{W}}_B^{\mathcal{C}}.
\end{split}
\end{equation}
To write the Lagrangian of \eqref{eq:newPrimal} we introduce, for $i=0,1$ and for all $x \in F^{[i]}$, the scalar dual variables $\lambda_x^{[i]} \geq 0$, and for each basis element $\overline{W}_k \in \overline{\W}^{\mathcal{C}}_B$, the scalar dual variable $\mu_k$. 
We also introduce, as slack variables, the scalar $\delta \geq 0$ and two real matrices $\Gamma^{[i]} \geq 0$. 
Writing $\bm{\lambda} := \{\lambda_x^{[i]}\}_{i=0,1; x \in F^{[i]}}$, $\bm{\Gamma} := \{\Gamma^{[i]}\}_{i=0,1}$ and $\bm{\mu} := \{\mu_k\}_k$, the Lagrangian is then a function of the primal, dual, and slack variables given as
\begin{align}
& \mathcal{L}(\varepsilon, W^{[0]}, W^{[1]}, \delta, \bm{\lambda}, \bm{\Gamma}, \bm{\mu})\notag\\ 
&= - (1 - \varepsilon) - \delta \varepsilon - \sum_{i=0}^1 \Tr[W^{[i]} \Gamma^{[i]}] + \sum_{i=0}^1 \sum_{x \in F^{[i]}} \lambda_x^{[i]} \Big(1 - \varepsilon - \Tr[W^{[i]} \mathsf{O}_x^{\otimes T}]\Big) + \sum_k \mu_k \Big( \Tr[(W^{[0]} + W^{[1]})\overline{W}_k] - 1\Big) \notag\\
&= - 1 + \varepsilon \Big(1 - \sum_{i=0}^{1} \sum_{x \in F^{[i]}} \lambda_x^{[i]} - \delta \Big) + \sum_{i=0}^{1} \Tr\Big[ W^{[i]}(-\Gamma^{[i]} - \sum_{x \in F^{[i]}} \lambda_x^{[i]}\mathsf{O}_x^{\otimes T} + \sum_{k}\mu_k \overline{W}_k) \Big] + \sum_{i=0}^{1} \sum_{x \in F^{[i]}} \lambda_x^{[i]} - \sum_k \mu_k.
\end{align}
Because any optimal solution $(\varepsilon^*, W^{[0]*}, W^{[1]*})$ of the primal SDP necessarily satisfies the constraints of \eqref{eq:newPrimal}, for any positive $(\delta, \bm{\lambda}, \bm{\Gamma})$ and any $\bm{\mu}$, we have $\mathcal{L}(\varepsilon^*, W^{[0]*}, W^{[1]*}, \delta, \bm{\lambda}, \bm{\Gamma}, \bm{\mu}) \leq - (1 - \varepsilon^*)$. 
Therefore, by considering the function obtained when minimising over all possible values of $(\varepsilon, W^{[0]}, W^{[1]})$ (even those not corresponding to feasible solutions of \eqref{eq:newPrimal}), one obtains a lower bound on the objective function $-(1 - \varepsilon^*)$ as
\begin{equation}
\label{eq:dual_lower}
g(\delta, \bm{\lambda}, \bm{\Gamma}, \bm{\mu}) = \min_{\varepsilon, W^{[0]}, W^{[1]}}\, \mathcal{L}(\varepsilon, W^{[0]}, W^{[1]}, \delta, \bm{\lambda}, \bm{\Gamma}, \bm{\mu}) \leq -(1 - \varepsilon^*).
\end{equation}
To obtain non-trivial lower bounds, i.e., finite values of $g$, certain conditions must be satisfied, as $g$ can be rewritten
\begin{equation}
g(\delta, \bm{\lambda}, \bm{\Gamma}, \bm{\mu}) = \begin{cases} \sum_{i} \sum_{x \in F^{[i]}} \lambda_x^{[i]} - 1 - \sum_k \mu_k & \text{if} 
\begin{cases} 1 - \sum_{i} \sum_{x \in F^{[i]}} \lambda_x^{[i]} - \delta = 0, \quad \text{and} \\  - \Gamma^{[i]} - \sum_{x \in F^{[i]}}\lambda^{[i]}_x \mathsf{O}_x^{\otimes T} + \sum_k \mu_k \overline{W}_k = 0, \, i \in \{0, 1\}, 
\end{cases} \\
- \infty & \text{otherwise.}
\end{cases}
\end{equation}

Then, the solution to the primal SDP $\varepsilon_T^{\mathcal{C}}(f)$ can be alternatively obtained by maximising this lower bound, which can be cast as the optimisation problem 
\begin{equation}
\begin{split}
 \max_{\delta, \bm{\lambda}, \bm{\Gamma}, \bm{\mu}} \ & \ \sum_{i=0}^{1} \sum_{x \in F^{[i]}} \lambda_x^{[i]} - 1 - \sum_k \mu_k, \\
\text{s.t.} \ & \sum_{i=0}^1 \sum_{x \in F^{[i]}} \lambda_x^{[i]} + \delta = 1, \\
       &  \Gamma^{[i]} + \sum_{x \in F^{[i]}} \lambda_x^{[i]} \mathsf{O}_x^{\otimes T} = \sum_k \mu_k \overline{W}_k, \, i \in \{0, 1\}, \\
       & \bm{\lambda} \geq 0,\ \bm{\Gamma} \geq 0,\ \delta \geq 0,
\end{split}
\end{equation}
where $\bm{\lambda} \geq 0$ means that $\lambda_x^{[i]} \geq 0$ for all $i=0,1$ and $x \in F^{[i]}$, similarly for $\bm{\Gamma} \geq 0$. 
This can be simplified through the removal of the slack variables $\bm{\Gamma}$ and $\delta$ to obtain
\begin{equation}
	\label{eq:sdp_dual_without_slack}
\begin{split}
 \min_{\bm{\lambda}, \bm{\mu}} \ & \ 1 + \sum_k \mu_k - \sum_{i=0}^{1} \sum_{x \in F^{[i]}} \lambda_x^{[i]}, \\
\text{s.t.} \ & \sum_{i=0}^1 \sum_{x \in F^{[i]}} \lambda_x^{[i]} \leq 1, \\
       & \sum_{x \in F^{[i]}} \lambda_x^{[i]} \mathsf{O}_x^{\otimes T} \leq \sum_k \mu_k \overline{W}_k, \, i \in \{0, 1\}, \\
       & \bm{\lambda} \geq 0.
\end{split}
\end{equation}
This SDP can be further simplified by writing $\nu = \sum_k \mu_k$ and $\overline{W} = \frac{1}{\nu}\sum_k \mu_k \overline{W}_k$, where $\overline{W}\in\overline{\mathcal{\W}}^\mathcal{C}$ since $\overline{\W}_B^{\mathcal{C}}=\{\overline{W}_k\}_k$ is an affine basis of this space.%
\footnote{One can easily check that $\nu \ge \sum_{x\in F^{[i]}}{\lambda_{x}^{[i]}} \ge 0$ by taking the trace of each side of the second constraint in~\eqref{eq:sdp_dual_without_slack} and recalling that $\bm{\lambda}\ge 0$. The case where $\nu = 0$ therefore corresponds to the trivial solution to the SDP. One can readily check that this trivial solution is also a feasible solution to~\eqref{eq:sdp_dual_nu} so that these two formulations of the problem are indeed equivalent.}
One then obtains
\begin{equation}
	\label{eq:sdp_dual_nu}
\begin{split}
\min_{\bm{\lambda},\nu, \overline{W}} \ & \ 1 +\nu - \sum_{i=0}^{1} \sum_{x \in F^{[i]}} \lambda_x^{[i]}, \\
\text{s.t.} \ & \sum_{i=0}^1 \sum_{x \in F^{[i]}} \lambda_x^{[i]}  \leq 1, \\
       & \sum_{x \in F^{[i]}} \lambda_x^{[i]} \mathsf{O}_x^{\otimes T} \leq\nu \overline{W}, \, i \in \{0, 1\}, \\
       & \bm{\lambda} \geq 0\\
       & \overline{W} \in \overline{\mathcal{W}}^\mathcal{C}.
\end{split}
\end{equation}
The above optimisation problem is no longer an SDP due to the nonlinear term $\nu\overline{W}$. 
However, one can absorb $\nu$ into the unnormalised operator $\overline{W}_* =\nu \overline{W}$, which then satisfies $\frac{1}{d^I}\Tr\overline{W}_* =\nu$.
Writing the cone generated by $\overline{\W}^\mathcal{C}$ as $\overline{\mathcal{W}}_*^\mathcal{C} = \{\nu \overline{W}\, :\, \nu \ge 0,\ \overline{W} \in \overline{\W}^\C \}$, one then obtains the SDP
\begin{equation}
\label{sdp_dual}
\begin{split}
\min_{\bm{\lambda}, \overline{W}_*} \ & \ 1 + \frac{1}{d^I} \Tr\overline{W}_* - \sum_{i=0}^{1} \sum_{x \in F^{[i]}} \lambda_x^{[i]}, \\
\text{s.t.} \ & \sum_{i=0}^1 \sum_{x \in F^{[i]}} \lambda_x^{[i]}  \leq 1, \\
       & \sum_{x \in F^{[i]}} \lambda_x^{[i]} \mathsf{O}_x^{\otimes T} \leq \overline{W}_*, \, i \in \{0, 1\}, \\
       & \bm{\lambda} \geq 0, \\
       & \overline{W}_* \in \overline{\W}_*^\mathcal{C},
\end{split}
\end{equation}
which is the dual of~\eqref{eq:newPrimal}.
This dual SDP minimises an objective function whose optimal value is $1 - \varepsilon_T^{\mathcal{C}}(f)$, and so for any solution of \eqref{sdp_dual} we have $\sum_{i=0}^{1} \sum_{x \in F^{[i]}} \lambda_x^{[i]} - \frac{1}{d^I} \Tr\overline{W}_* \leq \varepsilon_T^{\mathcal{C}}(f)$.
That is, any feasible solution to \eqref{sdp_dual} provides a lower bound on $\varepsilon_T^{\mathcal{C}}(f)$, while any feasible solution to the primal SDP \eqref{sdp} (or, equivalently, \eqref{eq:newPrimal}) provides an upper bound on $\varepsilon_T^{\mathcal{C}}(f)$.

Finally, let us give a characterisation of $\overline{\mathcal{W}}_*^\mathcal{C}$ for $\mathcal{C}\in\{\text{FO},\text{Gen}\}$.
Recall that $\overline{\mathcal{W}}^\mathcal{C}_*$ is the cone generated by the dual affine space $\overline{\mathcal{W}}^\mathcal{C}$ defined in Eq.~\eqref{eq:characDual}, and which, since $\Tr\overline{W}=d^I$ for $\overline{W}\in\overline{\W}^\mathcal{C}$, can be rewritten as
\begin{equation}
\label{eq:characDualCone}
	\overline{\mathcal{W}}_*^\mathcal{C} = \left\{ \overline{W}_* \in \L(H^{PA^{IO}_\mathcal{T}F}) \, : \, \Tr[\overline{W}_* \cdot {W}] = \frac{\Tr\overline{W}_*}{d^I},  \ \forall\, W \in \mathcal{W}^{\mathcal{C}} \right\}.
\end{equation}
Recalling also that any $W \in \mathcal{W}^\mathcal{C}$ is normalised so that $\Tr[\frac{1}{d^O} W] = 1$, where $d^O = d^P \prod_{k=1}^T d^O_k$ (see Appendix~\ref{annex:Gen}), we then obtain
\begin{equation}
	\overline{\mathcal{W}}_*^\mathcal{C} = \left\{ \overline{W}_* \in \L(H^{PA^{IO}_\mathcal{T}F}) \, : \, \Tr[\left(\overline{W}_* - \frac{\Tr(\overline{W}_*)}{d^I d^O}\id \right) \cdot {W}] = 0, \ \forall\, W \in \mathcal{W}^{\mathcal{C}} \right\}.
\end{equation}
This implies that $\overline{W}_* \in \overline{\mathcal{W}}_*^\mathcal{C}$ if and only if $\overline{W}_* - \frac{\Tr(\overline{W}_*)}{d^I d^O}\id$ is in the orthogonal complement of $\mathcal{W}^\mathcal{C}$ within the space of Hermitian operators.
 
Since $W\in\W^\mathcal{C}$ if and only if $W\in \mathcal{P}\cap \mathcal{L}^\mathcal{C}$ and $\Tr W = d^O$, we have $( \W^\mathcal{C})^\perp = \mathcal{P}^\perp + (\mathcal{L}^\mathcal{C})^\perp = (\mathcal{L}^\mathcal{C})^\perp$.
Therefore, taking $\Pi_{\mathcal{L}^\mathcal{C}}$ to be the projector onto $\mathcal{L}^\mathcal{C}$, the projector onto $\overline{\mathcal{W}}^\mathcal{C}_*$ is given by
\begin{align}
    \Pi_{\overline{\mathcal{W}}^\mathcal{C}_*}(\overline{W}_*) &= \overline{W}_* - \Pi_{\mathcal{L}^\mathcal{C}}\left(\overline{W}_* - \frac{\Tr(\overline{W}_*)}{d^I d^O}\id\right) \\
    &= \overline{W}_* - \Pi_{\mathcal{L}^\mathcal{C}}(\overline{W}_*) + \frac{\Tr(\overline{W}_*)}{d^I d^O}\id,
\end{align}
where the second line is obtained from the linearity of $\Pi_{\mathcal{L}^\mathcal{C}}$ and because $ \id \in \mathcal{L}^\mathcal{C}$. 
We thus have the characterisation
\begin{equation}
	\overline{\W}_*^\mathcal{C} = \{ \overline{W}_* \in \L(H^{PA^{IO}_\mathcal{T}F}) \, : \, \Pi_{\overline{\mathcal{W}}^\mathcal{C}_*}(\overline{W}_*) = \overline{W}_* \}.
\end{equation}


\section{Mathematical proofs from numerical results}
\label{appendix:extraction}

In this appendix, we describe two algorithms to extract strictly feasible solutions to the primal SDP~\eqref{sdp} and its dual~\eqref{sdp_dual} from numerical solutions to these SDPs.
This will allow us to extract analytical bounds on $\varepsilon_T^{\mathcal{C}}(f)$ and provide a rigorous mathematical proof of Theorem~\ref{theorem:gap}.

The nature of numerical SDP solvers means that the solutions they provide only satisfy the constraints of the SDP up to some numerical precision.
Since, strictly speaking, the constraints are not satisfied, we cannot directly conclude anything about the precision of the numerical result of the optimisation problem.
In Ref.~\cite{PhysRevLett.127.200504}, the authors propose an algorithm that takes a numerical (floating point), approximate solution of an SDP and returns an exact rational solution that is not specified with floating point numbers and rigorously satisfies the constraints of the SDP. 
While this alternative solution is no longer guaranteed to be optimal, it gives an upper or lower bound (depending on whether one maximises or minimises the objective function). 
By applying this approach to both the primal and dual forms of an SDP, one can obtain both lower and upper bounds on the true optimal solutions; i.e., an interval within which that solution is certified to lie.
Here we adapt the algorithm described in Ref.~\cite{PhysRevLett.127.200504} to our specific problem.

We start with an algorithm to extract an exact solution from the dual SDP~\eqref{sdp_dual} specified in Appendix~\ref{Appendix:dual}.
For a Boolean function $f$, this will provide a lower bound on $\varepsilon_T^{\mathcal{C}}(f)$, i.e., the minimal error with which one can compute $f$ using a fixed number $T$ of queries. 
A solution of the SDP~\eqref{sdp_dual} consists of a tuple $(\lambda^{[0]}, \lambda^{[1]}, S)$, where for $i=0, 1$, $\lambda^{[i]} = \{\lambda_x^{[i]} \}_{x \in F^{[i]}}$ and $S$ is a matrix in the cone $\overline{\W}^{\mathcal{C}}_*$ generated by the dual affine space $\overline{\W}^\mathcal{C}$ (see Appendix~\ref{Appendix:dual} and Ref.~\cite{PhysRevLett.127.200504}). 
The output of the following algorithm will be a tuple $(\lambda^{[0]}_{\text{final}}, \lambda^{[1]}_{\text{final}}, S_\text{final})$ with $\lambda^{[i]}_{\text{final}} = \{\lambda_{x,\text{final}}^{[i]}\}_{x \in F^{[i]}}$, which rigorously satisfies the constraints of the SDP~\eqref{sdp_dual}. 
Note that the output of the algorithm depends on the precision to which one rationalises the variables, which is a freely chosen parameter of the algorithm.
\begin{algorithm}[H] 
\caption{
}
\label{algorithm:dual}
\begin{enumerate}
    \item Define each $\lambda_{x,\text{frac}}^{[i]}$ as a rationalisation of $\lambda_x^{[i]}$, stored in an exact representation and where any negative $\lambda_x^{[i]}$ are set to 0.

    \item Define $\delta = 1 - \sum_{i=0}^1 \sum_{x \in F^{[i]}} \lambda_{x, \text{frac}}^{[i]}$ and let 
$
\lambda_{x, \text{final}}^{[i]} = \begin{cases} \lambda_{x, \text{frac}}^{[i]} + \delta/N & \text{if } \delta < 0 \text{ and } \lambda_{x, \text{frac}}^{[i]} \geq -\delta/N,   \\
\lambda_{x, \text{frac}}^{[i]} & \text{otherwise,} \end{cases}$
where $N\le 2^n$ is the maximal number of pairs $(i,x)$ such that $\lambda_{x,\text{frac}}^{[i]}\ge -\delta/N$.

	\item Define $S_{\text{frac}}$ as a rationalisation of $S$, stored in an exact representation.

	\item Define $S_{\text{Herm}} = \frac{1}{2} (S_{\text{frac}} + S_{\text{frac}}^\dagger)$.

	\item Define $S_{\text{valid}} = \Pi_{\overline{\W}^\mathcal{C}_*}(S_{\text{Herm}})$, where $\Pi_{\overline{\W}^\mathcal{C}_*}$ is the projection onto $\overline{\W}^{\mathcal{C}}_*$.

	\item Find the smallest $\mu$ such that $S_{\text{pos}} = D_{\mu}(S_\text{valid})$ is positive semidefinite, where  $D_{\mu}(S) = (1-\mu) S + \mu \id$.

	\item For $i=0, 1$, find the smallest non-negative $\eta_i$'s  such that $F_{\eta_i} (S_{\text{valid}})$ is positive semidefinite, where $F_{\eta_i}(S) = S + (\eta_i -1)\mathsf{O}^{[i]}$ with $\mathsf{O}^{[i]} = \sum_{x\in F^{[i]}} \lambda_{x, \text{final}}^{[i]} O_x^{\otimes T}$, and define $S_{\text{final}} = S + \eta_0 \mathsf{O}^{[0]} + \eta_1 \mathsf{O}^{[1]}$.
\end{enumerate}
\end{algorithm}

The first and second steps of the algorithm ensure that the final lambdas $\lambda_{x, \text{final}}^{[i]}$ are rational (stored, for example, as fractions that can be manipulated exactly), non-negative, and satisfy the constraints $\sum_{i=0}^1 \sum_{x \in F^{[i]}} \lambda_{x, \text{final}}^{[i]} \leq 1$. 
Steps 3 to 5 ensure that the matrix $\mathcal{S}_{\text{final}}$ is rational, Hermitian and belongs to $\overline{\W}^{\mathcal{C}}_*$.
Finally, Steps 6 and 7 ensure that, for $i \in \{0,1\}$, $S_{\text{final}} - \mathsf{O}^{[i]}$ is positive semidefinite.
Note that Step 6 is needed to ensure that such $\eta_i$ exist in Step 7, and that the minimisation in Step 7 typically needs to be performed numerically and thus approximately. This is not an issue, as one simply needs to find the smallest $\eta_i$ up to some desired precision that make the respective matrices $F_{\eta_i}(S)$ strictly positive semidefinite. 
Note as well that after Step 7, the matrix $S_{\text{final}}$ remains in $\overline{\W}^{\mathcal{C}}_*$, which follows from the linearity of $\Pi_{\overline{\W}^{\mathcal{C}}_*}$ and because, for any class $\mathcal{C} \in \{\text{FO, Gen}\}$, the $\mathsf{O}^{[i]}$'s are elements of $\overline{\W}^{\mathcal{C}}_*$ \cite{PhysRevLett.127.200504}.
The output of the algorithm $S_{\text{final}}$ is hence guaranteed to satisfy all the constraints of the dual SDP~\eqref{sdp_dual}. 
Furthermore, as stated in Appendix~\ref{Appendix:dual}, the dual SDP~\eqref{sdp_dual} minimises an objective function whose optimal is $1 - \varepsilon_T^\mathcal{C}(f)$, meaning that the output of Algorithm~\eqref{algorithm:dual} will provide a lower bound on $\varepsilon_T^\mathcal{C}(f)$, i.e.,
\begin{equation}
\sum_{i=0}^1 \sum_{x \in F^{[i]}} \lambda_{x, \text{final}}^{[i]} -  \frac{1}{d^I} \Tr[S_{\text{final}}] \leq \varepsilon^{C}_{T}(f). 
\end{equation}

The corresponding upper bound is obtained from a similar algorithm whose input is a solution of the primal SDP \eqref{sdp}, which is a superinstrument of two elements $\{W^{[i]}\}_{i=0,1}$ such that $W^{[0]} + W^{[1]} \in \mathcal{W}^{\mathcal{C}}$. Then, the output of the following algorithm will be a superinstrument $\{W_{\text{final}}^{[i]}\}_{i=0,1}$ that rigorously satisfies all the constraints of the primal SDP~\eqref{sdp}.

\begin{algorithm}[H]
    \caption{
    }
    \label{algorithm:primal}
    \begin{enumerate}
    \item Define $W^{[0]}_{\text{frac}}$ and  $W^{[1]}_{\text{frac}}$ as rationalisations of $W^{[0]}$ and $W^{[1]}$, stored in an exact representation.
    
\item For $i=0, 1$, define $W^{[i]}_{\text{Herm}} = \frac{1}{2} (W^{[i]}_{\text{frac}} + W_{\text{frac}}^{[i]\dagger})$.

\item  Define $W_{\text{proj}} = \Pi_{\L^\mathcal{C}}(W^{[0]}_{\text{Herm}} + W^{[1]}_{\text{Herm}})$, where $\Pi_{\L^\mathcal{C}}$ is the projection onto $\L^\mathcal{C}$.

\item  Define $W_{\text{corr}} = W_{\text{proj}} - W^{[0]}_{\text{Herm}} - W^{[1]}_{\text{Herm}}$.

\item  Define for $i=0,1$, $W^{[i]}_{\text{valid}} = W^{[i]}_{\text{Herm}} + \frac{1}{2}W_{\text{corr}}$.

\item  Find the smallest $\mu$ such that  for $i=0, 1$, $W^{[i]}_{\text{pos}} = D_{\mu} (W^{[i]}_\text{valid})$ is positive semidefinite, where $D_{\mu} (W) = (1-\mu) W + \mu \id$.

\item  Re-normalise, for $i=0,1$, $W_{\text{final}}^{[i]} = \frac{d^P \prod_{k=1}^T d^O_k}{\Tr[W^{[0]}_{\text{pos}} + W^{[i]}_{\text{pos}}]}W^{[i]}_{\text{pos}}$.
    \end{enumerate}
\end{algorithm}
The first and second steps of the algorithm make sure that the superinstrument elements $W_{\text{final}}^{[i]}$ are rational and Hermitian. 
Steps 3 to 5 ensure that their sum lies in the subspace $\L^{\mathcal{C}}$. 
Step 6 ensures that they are positive semidefinite, while Step 7 ensures that they are properly normalised. 
The output of the algorithm $\{W_{\text{final}}^{[i]}\}_{i\in\{0, 1\}}$ is then a superinstrument that satisfies all the constraints of the primal SDP~\eqref{sdp}. Because this SDP is a maximisation of the objective function $1 - \varepsilon$ with optimal value $1 - \varepsilon_T^{\mathcal{C}}(f)$, the output of Algorithm~\eqref{algorithm:primal} provides an upper bound on $\varepsilon_T^{\mathcal{C}}(f)$, i.e.,
\begin{equation}
\varepsilon^{C}_{T}(f) \leq \min_{i \in \{0, 1\}} \min_{x \in F^{[i]}} \Tr[W_{\text{final}}^{[i]} \mathsf{O}_x^{\otimes T}].
\end{equation}


\section{Numerical results}
\label{app:numerical}

In Table \ref{table:results} we summarise the numerical results obtained by solving the semidefinite program \eqref{sdp} for $n=4$ and $T=2$. The ID of a function $f$ corresponds to the integer obtained from its binary truth table, $\varepsilon_2^{\text{FO}}(f)$ is the numerical value corresponding to the minimum probability of error for which it can be computed using two queries with FO-supermaps and $\varepsilon_2^{\text{Gen}}(f)$ with general supermaps.
Table \ref{table:results} lists the 222 NPN representatives of Boolean functions on 4 input bits (note that it also contains functions that are constants or whose output depends on less than 4 input bits). 
For $\varepsilon_2^{\text{FO}}(f)$, our results coincide with those obtained in Ref.~\cite{Montanaro2015}, but we observe that for 179 representatives, $\varepsilon_2^{\text{Gen}}(f) < \varepsilon_2^{\text{FO}}(f)$, with a gap of 0.00947 (close to 1\%) for functions of Id: 5783, 5865 and 6630. 
To solve the SDPs we use the Matlab toolbox \texttt{Yalmip}~\cite{Lofberg2004} with the solver \texttt{SCS} \cite{scs};
our code is freely accessible on Github.%
\footnote{\url{https://github.com/pierrepocreau/QuantumQueryComplexity_ICO}}

\begin{table}
\footnotesize 
\caption{Numerical results for 2 queries and all Boolean functions of 4 input bits.}
\label{table:results}
\hspace*{-10mm}
\begin{minipage}{0.20\textwidth}
\begin{tabular}{|c|c|c|c|}
\hline
ID   & $\varepsilon_2^{\text{FO}}$ & $\varepsilon_2^{\text{Gen}}$ & Gap  \\ 
\hline
0     & 0.00000 & 0.00000 & 0.00000 \\
1     & 0.03846 & 0.03846 & 0.00000 \\
3     & 0.02000 & 0.02000 & 0.00000 \\
6     & 0.06897 & 0.06897 & 0.00000 \\
7     & 0.04620 & 0.04620 & 0.00000 \\
15    & 0.00000 & 0.00000 & 0.00000 \\
22    & 0.09380 & 0.09380 & 0.00000 \\
23    & 0.07409 & 0.07409 & 0.00000 \\
24    & 0.06897 & 0.06897 & 0.00000 \\
25    & 0.03957 & 0.03957 & 0.00000 \\
27    & 0.03475 & 0.03475 & 0.00000 \\
30    & 0.04411 & 0.04411 & 0.00000 \\
31    & 0.03043 & 0.02964 & 0.00080 \\
60    & 0.00000 & 0.00000 & 0.00000 \\
61    & 0.02383 & 0.02000 & 0.00383 \\
63    & 0.00000 & 0.00000 & 0.00000 \\
105   & 0.10000 & 0.10000 & 0.00000 \\
107   & 0.05936 & 0.05919 & 0.00016 \\
111   & 0.03254 & 0.02853 & 0.00400 \\
126   & 0.05263 & 0.05263 & 0.00000 \\
127   & 0.02858 & 0.02858 & 0.00000 \\
255   & 0.00000 & 0.00000 & 0.00000 \\
278   & 0.11611 & 0.11611 & 0.00000 \\
279   & 0.10061 & 0.10061 & 0.00000 \\
280   & 0.09380 & 0.09380 & 0.00000 \\
281   & 0.04411 & 0.04411 & 0.00000 \\
282   & 0.06387 & 0.06387 & 0.00000 \\
283   & 0.04136 & 0.04126 & 0.00010 \\
286   & 0.06842 & 0.06842 & 0.00000 \\
287   & 0.05475 & 0.05475 & 0.00000 \\
300   & 0.06387 & 0.06387 & 0.00000 \\
301   & 0.03637 & 0.03595 & 0.00041 \\
303   & 0.03436 & 0.02900 & 0.00536 \\
316   & 0.03846 & 0.03846 & 0.00000 \\
317   & 0.02000 & 0.02000 & 0.00000 \\
318   & 0.04408 & 0.04051 & 0.00357 \\
319   & 0.02807 & 0.02529 & 0.00279 \\
360   & 0.11611 & 0.11611 & 0.00000 \\
361   & 0.08382 & 0.08382 & 0.00000 \\
362   & 0.06842 & 0.06842 & 0.00000 \\
363   & 0.04510 & 0.04483 & 0.00027 \\
366   & 0.04408 & 0.04051 & 0.00357 \\
367   & 0.03147 & 0.02782 & 0.00364 \\
382   & 0.07739 & 0.07739 & 0.00000 \\
383   & 0.05410 & 0.05410 & 0.00000 \\
384   & 0.02000 & 0.02000 & 0.00000 \\
385   & 0.04620 & 0.04620 & 0.00000 \\
386   & 0.03957 & 0.03957 & 0.00000 \\
387   & 0.03475 & 0.03475 & 0.00000 \\
390   & 0.06387 & 0.06387 & 0.00000 \\
391   & 0.04136 & 0.04126 & 0.00010 \\
393   & 0.03475 & 0.03475 & 0.00000 \\
395   & 0.03679 & 0.03679 & 0.00000 \\
399   & 0.02000 & 0.02000 & 0.00000 \\
406   & 0.08382 & 0.08382 & 0.00000 \\
407   & 0.05985 & 0.05834 & 0.00151 \\
408   & 0.04411 & 0.04411 & 0.00000 \\
\vdots & \vdots & \vdots & \vdots \\ 
\hline
\end{tabular}
\end{minipage}
\hspace{6mm}
\begin{minipage}{0.20\textwidth}
\begin{tabular}{|c|c|c|c|}
\hline
ID   & $\varepsilon_2^{\text{FO}}$ & $\varepsilon_2^{\text{Gen}}$ & Gap \\
\hline
\vdots & \vdots & \vdots & \vdots \\ 
409   & 0.02383 & 0.02000 & 0.00383 \\
410   & 0.03637 & 0.03595 & 0.00041 \\
411   & 0.03079 & 0.03061 & 0.00018 \\
414   & 0.04510 & 0.04483 & 0.00027 \\
415   & 0.03527 & 0.03294 & 0.00233 \\
424   & 0.07409 & 0.07409 & 0.00000 \\
425   & 0.04136 & 0.04126 & 0.00010 \\
426   & 0.03043 & 0.02964 & 0.00080 \\
427   & 0.02000 & 0.02000 & 0.00000 \\
428   & 0.04136 & 0.04126 & 0.00010 \\
429   & 0.03079 & 0.03061 & 0.00018 \\
430   & 0.03436 & 0.02900 & 0.00536 \\
431   & 0.02351 & 0.02304 & 0.00047 \\
444   & 0.02000 & 0.02000 & 0.00000 \\
445   & 0.03499 & 0.03066 & 0.00433 \\
446   & 0.03147 & 0.02782 & 0.00364 \\
447   & 0.02000 & 0.02000 & 0.00000 \\
488   & 0.10061 & 0.10061 & 0.00000 \\
489   & 0.05985 & 0.05834 & 0.00151 \\
490   & 0.05475 & 0.05475 & 0.00000 \\
491   & 0.03527 & 0.03294 & 0.00233 \\
494   & 0.02807 & 0.02529 & 0.00279 \\
495   & 0.02000 & 0.02000 & 0.00000 \\
510   & 0.05410 & 0.05410 & 0.00000 \\
828   & 0.00000 & 0.00000 & 0.00000 \\
829   & 0.02521 & 0.02268 & 0.00253 \\
831   & 0.00000 & 0.00000 & 0.00000 \\
854   & 0.04508 & 0.04399 & 0.00110 \\
855   & 0.03109 & 0.02707 & 0.00402 \\
856   & 0.03637 & 0.03595 & 0.00041 \\
857   & 0.03880 & 0.03880 & 0.00000 \\
858   & 0.02000 & 0.02000 & 0.00000 \\
859   & 0.02652 & 0.02587 & 0.00065 \\
862   & 0.03430 & 0.03065 & 0.00366 \\
863   & 0.01396 & 0.01122 & 0.00274 \\
872   & 0.06842 & 0.06842 & 0.00000 \\
873   & 0.05375 & 0.05059 & 0.00316 \\
874   & 0.04508 & 0.04399 & 0.00110 \\
875   & 0.04884 & 0.04358 & 0.00526 \\
876   & 0.04408 & 0.04051 & 0.00357 \\
877   & 0.03870 & 0.03674 & 0.00196 \\
878   & 0.03430 & 0.03065 & 0.00366 \\
879   & 0.03434 & 0.03279 & 0.00155 \\
892   & 0.02521 & 0.02268 & 0.00253 \\
893   & 0.03399 & 0.03175 & 0.00224 \\
894   & 0.05258 & 0.05063 & 0.00194 \\
960   & 0.00000 & 0.00000 & 0.00000 \\
961   & 0.03043 & 0.02964 & 0.00080 \\
963   & 0.00000 & 0.00000 & 0.00000 \\
965   & 0.02000 & 0.02000 & 0.00000 \\
966   & 0.03436 & 0.02900 & 0.00536 \\
967   & 0.02351 & 0.02304 & 0.00047 \\
975   & 0.00000 & 0.00000 & 0.00000 \\
980   & 0.05475 & 0.05475 & 0.00000 \\
981   & 0.03109 & 0.02707 & 0.00402 \\
\vdots & \vdots & \vdots  & \vdots \\ 
\hline
\end{tabular}
\end{minipage}
\hspace{6mm}
\begin{minipage}{0.20\textwidth}
\begin{tabular}{|c|c|c|c|}
\hline
ID   & $\varepsilon_2^{\text{FO}}$ & $\varepsilon_2^{\text{Gen}}$ & Gap \\
\hline
\vdots & \vdots & \vdots & \vdots \\ 
982   & 0.04884 & 0.04358 & 0.00526 \\
983   & 0.03499 & 0.02990 & 0.00510 \\
984   & 0.03436 & 0.02900 & 0.00536 \\
985   & 0.02652 & 0.02587 & 0.00065 \\
985   & 0.02652 & 0.02587 & 0.00065 \\
987   & 0.03499 & 0.02990 & 0.00510 \\
988   & 0.02807 & 0.02529 & 0.00279 \\
989   & 0.01396 & 0.01122 & 0.00274 \\
990   & 0.03434 & 0.03279 & 0.00155 \\
1020  & 0.00000 & 0.00000 & 0.00000 \\
1632  & 0.00000 & 0.00000 & 0.00000 \\
1633  & 0.03846 & 0.03846 & 0.00000 \\
1634  & 0.02383 & 0.02000 & 0.00383 \\
1635  & 0.02000 & 0.02000 & 0.00000 \\
1638  & 0.05263 & 0.05263 & 0.00000 \\
1639  & 0.03499 & 0.03066 & 0.00433 \\
1641  & 0.06387 & 0.06387 & 0.00000 \\
1643  & 0.03449 & 0.03288 & 0.00162 \\
1647  & 0.00000 & 0.00000 & 0.00000 \\
1650  & 0.03079 & 0.03061 & 0.00018 \\
1651  & 0.02652 & 0.02587 & 0.00065 \\
1654  & 0.03499 & 0.03066 & 0.00433 \\
1656  & 0.04408 & 0.04051 & 0.00357 \\
1657  & 0.03568 & 0.02975 & 0.00593 \\
1658  & 0.03430 & 0.03065 & 0.00366 \\
1659  & 0.03836 & 0.03714 & 0.00122 \\
1662  & 0.04641 & 0.04247 & 0.00394 \\
1680  & 0.10000 & 0.10000 & 0.00000 \\
1681  & 0.08382 & 0.08382 & 0.00000 \\
1683  & 0.05375 & 0.05059 & 0.00316 \\
1686  & 0.06387 & 0.06387 & 0.00000 \\
1687  & 0.03568 & 0.02975 & 0.00593 \\
1695  & 0.00000 & 0.00000 & 0.00000 \\
1712  & 0.05936 & 0.05919 & 0.00016 \\
1713  & 0.04510 & 0.04483 & 0.00027 \\
1714  & 0.05985 & 0.05834 & 0.00151 \\
1715  & 0.04884 & 0.04358 & 0.00526 \\
1716  & 0.04510 & 0.04483 & 0.00027 \\
1717  & 0.03870 & 0.03674 & 0.00196 \\
1718  & 0.03449 & 0.03288 & 0.00162 \\
1719  & 0.03836 & 0.03714 & 0.00122 \\
1721  & 0.03568 & 0.02975 & 0.00593 \\
1725  & 0.04491 & 0.04053 & 0.00438 \\
1776  & 0.03254 & 0.02853 & 0.00400 \\
1777  & 0.03147 & 0.02782 & 0.00364 \\
1778  & 0.03527 & 0.03294 & 0.00233 \\
1782  & 0.00000 & 0.00000 & 0.00000 \\
1785  & 0.00000 & 0.00000 & 0.00000 \\
1910  & 0.04641 & 0.04247 & 0.00394 \\
1912  & 0.07739 & 0.07739 & 0.00000 \\
1913  & 0.04491 & 0.04053 & 0.00438 \\
1914  & 0.05258 & 0.05063 & 0.00194 \\
1918  & 0.07790 & 0.07790 & 0.00000 \\
1968  & 0.03254 & 0.02853 & 0.00400 \\
1969  & 0.03527 & 0.03294 & 0.00233 \\
\vdots & \vdots & \vdots & \vdots \\ 
\hline
\end{tabular}
\end{minipage}
\hspace{7mm}
\begin{minipage}{0.20\textwidth}
\begin{tabular}{|c|c|c|c|}
\hline
ID   & $\varepsilon_2^{\text{FO}}$ & $\varepsilon_2^{\text{Gen}}$ & Gap \\
\hline
\vdots & \vdots & \vdots & \vdots \\ 
1972  & 0.03147 & 0.02782 & 0.00364 \\
1973  & 0.03434 & 0.03279 & 0.00155 \\
1974  & 0.03836 & 0.03714 & 0.00122 \\
1980  & 0.03399 & 0.03175 & 0.00224 \\
2016  & 0.00000 & 0.00000 & 0.00000 \\
2017  & 0.02807 & 0.02529 & 0.00279 \\
2018  & 0.02351 & 0.02304 & 0.00047 \\
2019  & 0.01396 & 0.01122 & 0.00274 \\
2022  & 0.03499 & 0.02990 & 0.00510 \\
2025  & 0.03399 & 0.03175 & 0.00224 \\
2032  & 0.02858 & 0.02858 & 0.00000 \\
2033  & 0.02000 & 0.02000 & 0.00000 \\
2034  & 0.02000 & 0.02000 & 0.00000 \\
2040  & 0.05410 & 0.05410 & 0.00000 \\
4080  & 0.00000 & 0.00000 & 0.00000 \\
5736  & 0.00000 & 0.00000 & 0.00000 \\
5737  & 0.03846 & 0.03846 & 0.00000 \\
5738  & 0.02521 & 0.02268 & 0.00253 \\
5739  & 0.02000 & 0.02000 & 0.00000 \\
5742  & 0.05258 & 0.05063 & 0.00194 \\
5758  & 0.07790 & 0.07790 & 0.00000 \\
5761  & 0.14045 & 0.14045 & 0.00000 \\
5763  & 0.09363 & 0.09363 & 0.00000 \\
5766  & 0.06387 & 0.06387 & 0.00000 \\
5767  & 0.06735 & 0.06296 & 0.00440 \\
5769  & 0.09363 & 0.09363 & 0.00000 \\
5771  & 0.05363 & 0.05341 & 0.00022 \\
5774  & 0.03449 & 0.03288 & 0.00162 \\
5782  & 0.03846 & 0.03846 & 0.00000 \\
\textbf{5783}  & \textbf{0.04647} & \textbf{0.03700} & \textbf{0.00947} \\
5784  & 0.05375 & 0.05059 & 0.00316 \\
5785  & 0.06735 & 0.06296 & 0.00440 \\
5786  & 0.03568 & 0.02975 & 0.00593 \\
5787  & 0.04518 & 0.04286 & 0.00232 \\
5790  & 0.02000 & 0.02000 & 0.00000 \\
5801  & 0.06735 & 0.06296 & 0.00440 \\
5804  & 0.03870 & 0.03674 & 0.00196 \\
5805  & 0.04518 & 0.04286 & 0.00232 \\
5820  & 0.04491 & 0.04053 & 0.00438 \\
\textbf{5865}  & \textbf{0.04647} & \textbf{0.03700} & \textbf{0.00947} \\
6014  & 0.12571 & 0.12571 & 0.00000 \\
6030  & 0.00000 & 0.00000 & 0.00000 \\
6038  & 0.02000 & 0.02000 & 0.00000 \\
6040  & 0.04884 & 0.04358 & 0.00526 \\
6042  & 0.03836 & 0.03714 & 0.00122 \\
6060  & 0.03434 & 0.03279 & 0.00155 \\
6120  & 0.00000 & 0.00000 & 0.00000 \\
6375  & 0.10000 & 0.10000 & 0.00000 \\
6625  & 0.05363 & 0.05341 & 0.00022 \\
6627  & 0.04518 & 0.04286 & 0.00232 \\
\textbf{6630}  & \textbf{0.04647} & \textbf{0.03700} & \textbf{0.00947} \\
7128  & 0.00000 & 0.00000 & 0.00000 \\
7140  & 0.00000 & 0.00000 & 0.00000 \\
7905  & 0.10000 & 0.10000 & 0.00000 \\
15555 & 0.00000 & 0.00000 & 0.00000 \\
27030 & 0.00000 & 0.00000 & 0.00000 \\
\hline
\end{tabular}
\end{minipage}
\end{table}
\end{document}